\documentclass[a4paper,11pt]{article}

\usepackage[affil-it]{authblk}
\usepackage{hyperref}
\usepackage{braket}
\usepackage[utf8]{inputenc}

\usepackage{amsmath,amssymb}
\usepackage{esint}
\usepackage{bm}
\usepackage{graphicx}
\usepackage{caption}
\usepackage{subcaption}
\usepackage[margin=0.9in]{geometry}
\usepackage[all]{xy}
\usepackage{algorithm}
\usepackage{amsthm}
\usepackage{cite}
\usepackage{color}

\input{Qcircuit}

\newtheorem{theorem}{Theorem}

\newtheorem{lemma}{Lemma}

\newtheorem{proposition}{Proposition}

\newtheorem{definition}{Definition}

\newcommand{\cP}{{\sf P}}

\newcommand{\cBPP}{{\sf BPP}}
\newcommand{\cNP}{{\sf NP}}
\newcommand{\csP}{{\sf \texttt{\#} P}}
\newcommand{\cMA}{{\sf MA}}

\newcommand{\ccoNP}{{\sf coNP}}
\newcommand{\cBQP}{{\sf BQP}}

\newcommand{\cPromiseQMA}{{\sf PromiseQMA}}

\newcommand{\cQMA}{{\sf QMA}}
\newcommand{\cQCMA}{{\sf QCMA}}

\newcommand{\ccoQCMA}{{\sf coQCMA}}
\newcommand{\cPpoly}{{\sf P/poly}}
\newcommand{\cNPpoly}{{\sf NP/poly}}
\newcommand{\cMApoly}{{\sf MA/poly}}
\newcommand{\ccoNPpoly}{{\sf coNP/poly}}
\newcommand{\cPnd}{{\sf P/O(n^d)}}
\newcommand{\cNPnd}{{\sf MA/O(n^d)}}
\newcommand{\ccoNPnd}{{\sf coMA/O(n^d)}}
\newcommand{\cBPPrpoly}{{\sf BPP/rpoly}}
\newcommand{\cMArpoly}{{\sf MA/rpoly}}

\newcommand{\cQCMAqpoly}{{\sf QCMA/qpoly}}
\newcommand{\ccoQCMAqpoly}{{\sf coQCMA/qpoly}}
\newcommand{\cPH}{{\sf PH}}
\newcommand{\cIP}{{\sf IP}}
\newcommand{\cPSPACE}{{\sf PSPACE}}

\newcommand{\gesdec}{{\sf NP/poly} \cap {\sf coNP/poly}}

\newcommand{\qgesdec}{{\sf QCMA/qpoly} \cap {\sf coQCMA/qpoly}}

\begin{document}
\sloppy
\title{Complexity-theoretic limitations on blind \\ delegated quantum computation}
\author[1]{Scott Aaronson\footnote{Email: aaronson@cs.utexas.edu}}
\author[2]{Alexandru Cojocaru\footnote{Email: a.cojocaru@sms.ed.ac.uk}}
\author[2,3]{Alexandru Gheorghiu\footnote{Email: andrugh@caltech.edu}}
\author[2,4]{Elham Kashefi\footnote{Email: ekashefi@inf.ed.ac.uk}}
\affil[1]{Department of Computer Science, University of Texas at Austin}
\affil[2]{School of Informatics, University of Edinburgh}
\affil[3]{Department of Computing and Mathematical Sciences, California Institute of Technology}
\affil[4]{CNRS LIP6, Universit\'{e} Pierre et Marie Curie, Paris}

\date{}
\maketitle

\begin{abstract}
Blind delegation protocols allow a client to delegate a computation to a server so that the server learns nothing about the input to the computation apart from its size.
For the specific case of \emph{quantum computation} we know, from work over the past decade, that blind delegation protocols can achieve information-theoretic security (provided the client and the server exchange some amount of quantum information). In this paper we prove, provided certain complexity-theoretic conjectures are true, that the power of \emph{information-theoretically secure} blind delegation protocols for quantum computation (ITS-BQC protocols) is in a number of ways constrained. 

In the first part of our paper we provide some indication that ITS-BQC protocols for delegating polynomial-time quantum computations in which the client and the server interact only classically are unlikely to exist.
We first show that having such a protocol in which the client and the server exchange $O(n^d)$ bits of communication, implies that $\cBQP \subset \cNPnd$. We conjecture that this containment is unlikely by proving that there exists an oracle relative to which $\cBQP \not\subset \cNPnd$.
We then show that if an ITS-BQC protocol exists in which the client and the server interact only classically and which allows the client to delegate quantum sampling problems to the server (such as $\textsc{BosonSampling}$) then there exist non-uniform circuits of size $2^{n - \mathsf{\Omega}(n/log(n))}$, making polynomially-sized queries to an $\cNP^{\cNP}$ oracle, for computing the permanent of an $n \times n$ matrix.

The second part of our paper concerns ITS-BQC protocols in which the client and the server engage in one round of quantum communication and then exchange polynomially many classical messages. First, we provide a complexity-theoretic upper bound on the types of functions that could be delegated in such a protocol by showing that they must be contained in $\qgesdec$. Then, we show that having such a protocol for delegating $\cNP$-hard functions implies $\mathsf{coNP^{NP^{NP}}} \subseteq \cNP^{\cNP^{\cPromiseQMA}}$.
\end{abstract}

\clearpage


\clearpage

\section{Introduction}
An important area of research in modern cryptography is that of performing \emph{computations on encrypted data}. The general idea is that a client wants to compute some function $f$ on some input $x$, but lacks the computational power to do this in a reasonable amount of time. Luckily, the client has access to a computationally powerful server (cloud, cluster etc) which can compute $f(x)$ quickly. However, because the computation might involve sensitive or classified information, or the server could be compromised remotely, we would like the input $x$ to be hidden from the server at all times. The client can simply encrypt $x$, but this raises the question: how can the server compute $f(x)$ if it doesn't know $x$?
The general problem of computing on encrypted data was first considered by Rivest, Adleman and Dertouzos \cite{rad}. Since then, instances of this problem have appeared in many areas of modern research including those of electronic voting, machine learning on encrypted data, program obfuscation and others \cite{example1, example2, example3, example4, example5, example6}.

It was shown in 2009, when Gentry produced the first \emph{fully homomorphic encryption scheme}, that performing classical computations on encrypted data is possible \cite{gentry}. In homomorphic encryption the client has a pair of efficient algorithms $(Enc, Dec)$, which respectively perform encryption and decryption, and which satisfy the property $Dec(f, x, Eval(f, Enc(x))) = f(x)$, for any function $f$ from some set $\mathcal{C}$. In other words,
the server evaluates $f$ on the encrypted input $Enc(x)$ using $Eval$ and returns this to the client which can then decrypt it to $f(x)$. Of course, the server should not be able to infer information about $x$ from $Enc(x)$, a condition which is typically expressed through the criterion of \emph{semantic security} \cite{crypto}.
If the set $\mathcal{C}$ contains all polynomial-sized circuits then the scheme becomes a fully homomorphic encryption scheme, commonly abbreviated FHE. All known FHE schemes are secure under \emph{cryptographic assumptions}.

Computing on encrypted data becomes particularly interesting when the server is a \emph{quantum computer}. This is because efficient quantum algorithms have been found for various problems which are believed to be intractable for classical computers. In fact, it has been shown that if
a classical computer and a quantum computer are both given black-box or oracle access to certain functions, then the quantum computer exponentially outperforms the classical computer \cite{bv, simon, niel, aaronson}.
Classical clients would therefore be highly motivated to delegate problems to quantum computers. However, ensuring the privacy of their inputs is challenging. In particular, we'd have to solve the following problems:
\begin{itemize}
\item Devise an encryption scheme which is secure against quantum computers and does not leak information to the server about the client's input.
\item Ensure that the encryption scheme allows the client to recover the output of the computation from the result provided by the quantum server.
\item Ensure that the protocol is efficient for the client. Ideally, the number of rounds of interaction between the client and the server as well as the client's local computations, should scale at most polynomially with the size of the input.
\end{itemize}
In spite of these stringent requirements, protocols that achieve these properties already exist and are known collectively as \emph{delegated blind quantum computing schemes} \cite{fitzsimons2016private}.
In such protocols, a probabilistic polynomial-time client is able to delegate polynomial-time quantum computations to a server in such a way that the client's input (apart from an upper bound on its size) is kept hidden from the server in an \emph{information-theoretic} sense. All of the above schemes require the client and the server to share at least one round of quantum communication.
\emph{Universal Blind Quantum Computation} (UBQC), shown schematically in Figure~\ref{fig:ubqc}, is an example of such a protocol \cite{bfk}.

\begin{figure}[htbp!]
\centering
\includegraphics[scale=0.35]{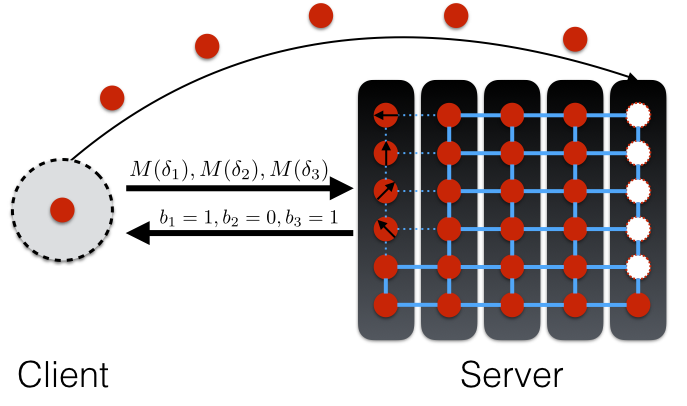}
\caption{Universal Blind Quantum Computation \cite{bfk}. In UBQC, a classical client augmented with the ability to prepare single-qubit states sends these qubits to the server along with instructions on how to entangle and measure them in order to perform a computation. The $M(\delta_i)$ indicate measurement instructions and the $b_i$ indicate the server's responses for these instructions (if he follows the protocol, these responses would represent the outcomes of the measurements that the client instructed him to perform).}
\label{fig:ubqc}
\end{figure}
The first blind delegation protocol was devised by Childs in \cite{Childs:2005:SAQ:2011670.2011674}, and since then these protocols have been improved and extended in various works \cite{PhysRevA.87.050301, PhysRevLett.111.230501, 1607.00758, fk, PhysRevLett.111.230502, Morimae:2015:GSB:2871393.2871395, 1606.06931, 1703.03754}. UBQC and related protocols require the client and the server to exchange only one quantum message, while the rest of the communication is classical \cite{bfk, abe, doi:10.1139/cjp-2015-0030}. This quantum message (which is sent by the client to the server) consists of a tensor product of single-qubit states. As such, the only quantum capability the client needs is the ability to prepare single-qubit states.

\noindent In this paper, we explore two questions pertaining to blind delegation protocols:
\begin{itemize}
\item[\textbf{(1)}] Is there a scheme for blind quantum computing that is
information-theoretically secure, and that requires only classical
communication between client and server?
\item[\textbf{(2)}] For schemes in which the client and the server are allowed one round of quantum communication, which functions can the client delegate to the server while maintaining information-theoretic security? In particular, could the client delegate the evaluation of $\cNP$-hard functions?
\end{itemize}
We provide some indication, based on complexity-theoretic conjectures, that the answer to the first question is no. In other words, provided these complexity-theoretic conjectures hold, a classical client running in polynomial time and communicating only classically with a server cannot delegate arbitrary polynomial-time quantum computations to that server while keeping its input hidden in an information-theoretic sense. Importantly, our result does not contradict recent results on quantum fully homomorphic encryption with a classical client \cite{mahadev2017classical, cryptoeprint:2018:338}, since those schemes are based on cryptographic assumptions: \emph{we are interested only in information-theoretic security}.

In answer to the second question, we provide a complexity-theoretic upper bound on the types of functions that can be evaluated by UBQC-type protocols (i.e. protocols in which the client can send one quantum message to the server\footnote{In fact our result concerns protocols in which the client and the server start with \emph{one round} of quantum communication, followed by polynomially-many rounds of classical communication. In other words, not only is there one quantum message from the client to the server, but the server is also allowed to respond with a quantum message.}). We show that, under plausible complexity-theoretic assumptions, this upper bound prevents the client from delegating $\cNP$-hard functions to the server. Thus, allowing for quantum communication between the client and the server expands the set of functions that the client can delegate to the server to include $\cBQP$, but not enough so as to include $\cNP$ as well.

\subsection{Main results} \label{subsect:main}
We phrase our results formally using the concept of a \emph{generalised encryption scheme} (GES) introduced by Abadi, Feigenbaum and Killian \cite{afk}, which is defined in Section~\ref{sect:crypto}.
Roughly speaking, a GES is a protocol between a probabilistic polynomial-time classical client and a computationally unbounded server for computing on encrypted data. The client sends the server a description of some function\footnote{Unless otherwise specified, we restrict our attention to decision problems. This is why the function $f$ has the codomain $\{0, 1\}$.} $f : \{0, 1\}^n \rightarrow \{0, 1\}$.
Using some polynomial-time algorithm denoted $E$, the client encrypts its input $x$, and sends $E(x)$ to the server.
The server and the client then interact for a number of rounds which is polynomial in the length of $x$.
Finally, using a polynomial-time decryption algorithm denoted $D$, the client decrypts the server's responses and obtains $f(x)$ with probability $1/2 + 1/\operatorname*{poly}(n)$.
Importantly throughout the protocol, the server learns no more than the length of $x$. Because the server is computationally unbounded, the scheme requires information-theoretic security. Abadi et al.\ gave a complexity theoretic upper bound on the types of functions that admit such a scheme.
They showed that any function $f$ that the client could delegate in a GES must be contained in the class $\gesdec$. We give a simplified proof of this result in Section~\ref{sect:crypto}.

The GES framework allows us to restate the questions we address in this paper as follows:
\begin{itemize}
\item[\textbf{(1)}] Can we design a GES for delegating $\cBQP$ functions? Note that, by the Abadi et al.\ result, this is the same as asking whether $\cBQP \subset \gesdec$. We will consider two variants on the GES framework: one which allows the client to delegate sampling problems to the server, and one in which the total communication between client and server is bounded by $\mathsf{O(n^d)}$, for some constant $d > 0$. For the former, we show that having such a scheme for quantum sampling problems, like $\textsc{BosonSampling}$, implies that circuits exist which can compute the permanent of a matrix more efficiently than we believe is possible. For the latter, having a GES with bounded communication for polynomial-time quantum computation implies that $\cBQP \subset \cNPnd$, and we argue that this containment is unlikely by providing an oracle separation between these classes.
\item[\textbf{(2)}] If we change the GES framework to allow one round of quantum communication between the client and the server, what functions can the client delegate to the server? We answer this question by ``quantising'' the Abadi et al.\ result and showing that such functions would be contained in $\qgesdec$ (a quantum analogue of $\gesdec$). We also show that $\qgesdec$ is unlikely to contain $\cNP$-hard functions.  
\end{itemize}

\subsubsection{Generalised encryption scheme for \textsf{BQP} decision problems}
As we have mentioned, for the case of decision problems, Abadi et al.\ showed that the class of problems that a client can delegate to a server using the GES framework is contained in $\gesdec$. They also observed that if $\cNP$-hard functions could be delegated by the client using a GES, then $\cNP \subset \gesdec$, and, in particular, $\cNP \subset \ccoNPpoly$. Yap showed that having such a containment leads to a collapse of the polynomial hierarchy at the third level \cite{yap}. In other words, it seems unlikely that $\cNP$-hard problems would admit a GES.

What about $\cBQP$-hard functions? The Abadi et al.\ result implies that having a GES for $\cBQP$-hard functions leads to $\cBQP \subset \gesdec$.
While we would like to argue, similarly, that such a containment leads to a collapse of the polynomial hierarchy, even $\cBQP = \cP$ isn't known to lead to such a collapse. We instead consider a modified GES in which the total communication between the client and the server is upper bounded by a polynomial of \emph{fixed} degree, $d > 0$, in the size of the input\footnote{Note that we impose no such restriction on the running time of the client.}. In that case, it can be shown that $\cBQP \subset \cNPnd \cap \ccoNPnd$ (see the proof of Theorem~\ref{thm:afk}). We argue that this containment is unlikely based on the following result, which we prove in Section~\ref{sect:proofs1}:
\begin{theorem} \label{thm:noGESBQP}
For each $d \in \mathbb{N}$, there exists an oracle $O_d$ such that $\cBQP^{O_d}$ is not contained in $(\cNPnd)^{O_d}$.
\end{theorem}

Essentially, the theorem shows that relative to an oracle $O_d$, there are problems that can be solved by a polynomial-time quantum algorithm, but which a classical client cannot delegate to a server in a GES with bounded communication. Since the oracle is parameterised by $d$, we are in fact defining a family of oracles. The specific problem on which the oracle $O_d$ is based is a version of \emph{Simon's problem} \cite{simon}.
Simon's problem is the following: for an input of size $n$, and given oracle access to a function $g : \{0, 1\}^n \rightarrow \{0, 1\}^n$ that is guaranteed to be either an injective function, or a $2$-to-$1$ and periodic function\footnote{In other words, there exists a period $s \in \{0, 1\}^n$, $s \neq 0^n$, such that for all $x, y \in \{0, 1\}^n$, $x \neq y$, it is the case that $g(x) = g(y)$ iff. $x = s \oplus y$.}, the task is to decide which is the case.
Simon provided a polynomial-time quantum algorithm for solving this problem, thus showing that it belongs to $\cBQP$ (relative to the function oracle).
For the case in which one should accept when the function is $2$-to-$1$, the problem can be shown to be outside of $\cMA$ (relative to the function oracle). As such, Simon's problem provides an oracle separation between $\cBQP$ and $\cMA$.

In Simon's original construction, the oracle function is the same for all inputs of size $n$. Note that,
this version of the problem can be solved with one bit of advice: for all inputs of size $n$, the advice bit simply specifies whether the function is $1$-to-$1$ or $2$-to-$1$ and periodic. Therefore such a setup would not be useful in our case.
For this reason, in our proof of Theorem~\ref{thm:noGESBQP}, the function that the oracle provides access to is input-dependent.
The problem we define, relative to this oracle, is again to decide whether the function is $1$-to-$1$ or the function is $2$-to-$1$ and periodic.
However, we can show that, by considering a sufficiently large domain for these functions --- in other words, by letting $g: \{0, 1\}^{n^D} \rightarrow \{0, 1\}^{n^D}$ for some $D > d$ --- the problem is not contained in $(\cNPnd)^{O_d}$, but is nevertheless contained in $\cBQP$. The proof uses a diagonalisation argument and can be found in Section~\ref{sect:proofs1}.

Unfortunately, the same oracle cannot be used to separate $\cBQP$ from $\cNPpoly$. This is because $D$ is a function of $d$; to prove a separation with respect to $\cNPpoly$, where the length of the advice string can be any polynomial, we would have to find an oracle that works \emph{for all} possible values of $d$. 
It would be interesting to see whether the oracle that Raz and Tal \cite{raztal} recently used to prove a separation between $\cBQP$ and $\cPH$ could also be used in order to separate $\cBQP$ from $\cNPpoly$, or even from $\mathsf{PH/poly}$.
We leave this as an open problem.

One can argue that oracle results do not constitute compelling evidence on the relationships between complexity classes. For example, it has been known for a while that there exist oracles $O_1, O_2$ such that $\cP^{O_1} \neq \cNP^{O_1}$ but $\cP^{O_2} = \cNP^{O_2}$, and that, while $\cIP = \cPSPACE$, there is an oracle such that $\cIP^O \neq \cPSPACE^O$. Nonetheless, oracles allow us to study the query complexity of problems in different models of computation.
In fact, there are situations in practice where computer programs are restricted to making black-box calls to functions in order to determine their properties \cite{blackbox}. Apart from this, oracle results have also inspired a number of important developments in algorithms and complexity theory\footnote{A notable example is the fact that Simon's oracle separation between $\cBPP$ and $\cBQP$ led to Shor's algorithm for factoring and computing the discrete logarithm \cite{shor}}.
For more arguments concerning the usefulness of oracle results, see Section 1.3 of \cite{aaronsonph}.

\subsubsection{Generalised encryption scheme for \textsf{BQP} sampling problems}
We consider what would happen if we have a generalised encryption scheme which allowed a client to delegate a sampling problem, such as \textsc{BosonSampling}, to the server.
\textsc{BosonSampling}, defined by Aaronson and Arkhipov in \cite{bosonsampling}, is essentially the problem of simulating the statistics of photons (bosons) passing through a linear optics network. One starts with a configuration of identical photons in known locations (referred to as \emph{modes}). The photons then pass through the linear optics network, which consists of optical elements (beamsplitters and phase shifters). Finally, one performs a measurement to determine the new locations of the photons in the \emph{output modes} of the system. The reason this is referred to as a sampling problem is because we have a probability distribution over the different configurations of photons in the output modes. In \emph{exact} \textsc{BosonSampling}, which is the problem we consider, the task is to produce a sample from that probability distribution.
Aaronson and Arkhipov showed that the probability of observing a particular configuration of photons is proportional to the squared permanent of a matrix that can be obtained efficiently from the description of the optical network. They also showed that no polynomial-time probabilistic algorithm can sample from this distribution, unless the polynomial hierarchy collapses at the third level \cite{bosonsampling}. As such, while a quantum computer could simulate the optical network and sample from the target distribution in polynomial time, it seems unlikely that classical computers could do the same.

In a GES for exact \textsc{BosonSampling}, the client's input would be a description of a linear optics network\footnote{In principle, one could also specify the configuration of the photons in the input modes as part of the client's input. Equivalently, however, one can always initialise the input modes to some fixed initial state, and produce whichever starting state is in fact desired by altering the linear optics network.}. The client would like to delegate to the server the task of sampling from the \textsc{BosonSampling} distribution associated with this network, while keeping the description of the network hidden from the server. In other words, upon interacting with the server and decrypting its responses, the client should obtain a sample from the \textsc{BosonSampling} distribution. At the same time, the server learns at most an upper bound on the size of the network. We show the following:

\begin{theorem} \label{thm:noGESsampBQP}
If exact \textsc{BosonSampling} admits a GES, then for any matrix $X \in \{-1, 0, 1 \}^{n \times n}$, there exist circuits of size $2^{n - \mathsf{\Omega} \left(\frac{n}{\log n} \right)}$, making polynomially-sized queries to an $\cNP^\cNP$ oracle, for computing the permanent of $X$.
\end{theorem}

Computing the permanent of a matrix is a problem known to be $\csP$-hard. By Toda's theorem, this means that if computing the permanent were possible at any level of the polynomial hierarchy, the hierarchy would collapse at that level \cite{toda}. Moreover, the best known algorithm for computing the permanent, by Bj{\"o}rklund, has a run-time of $2^{\mathsf{n - \Omega\left(\sqrt{n/log(n)}\right)}}$ \cite{andreas}. Prior to that, the leading algorithm for computing the permanent was Ryser's algorithm, developed over $50$ years ago, which requires $\mathsf{O(n2^{n})}$ arithmetic operations \cite{ryser1963combinatorial}. We conjecture that the circuits of Theorem~\ref{thm:noGESsampBQP} do not exist and, thus, that there can be no GES for \textsc{BosonSampling}. The proof of this result can be found in Section~\ref{sect:proofs2}.

\subsubsection{Quantum generalised encryption scheme}
While having a GES for delegating $\cBQP$ computations seems unlikely, we know that giving the client some minimal quantum capabilities removes this limitation: schemes such as UBQC exist which allow for the information-theoretically secure blind delegation of quantum computations.
In the spirit of the Abadi et al.\ result, it is natural to consider \emph{quantum generalised encryption schemes} (or QGES), in which the client is no longer classical, and investigate the complexity-theoretic upper bounds on functions that admit such a protocol. For the QGES, we are still assuming information-theoretic security and that the encryption scheme leaks at most the size of the input. However, unlike the GES, the client is now assumed to be a quantum computer performing polynomial-time computations\footnote{It should be noted that our upper bound on the power of QGES schemes also holds if the client is restricted to $\cBPP$ computations (as is the case in UBQC), since $\cBPP \subseteq \cBQP$.}. Additionally, the client and the server perform one round of quantum communication at the beginning of the protocol. The rest of the communication is classical.

We impose one other restriction on the QGES, known as \emph{offline-ness}. Roughly speaking, an offline protocol is one in which the client does not need to commit to any particular input (of a given size), after having sent the quantum message to the server. The quantum message only depends on the size of the input. We note that offline-ness is a property which UBQC and all other currently known blind quantum computing protocols share.
From a practical perspective, this presents the client with the option of sending the first quantum message to the server and deciding at a later time on which input the server should perform the computation. One could imagine, for instance, that the client and the server have access to a quantum channel for a limited amount of time. In practice, such a situation can occur if the communication between the parties is mediated by a satellite, as is the case with satellite-based quantum-key distribution \cite{liao2017satellite}. In this case, the satellite is in the line of sight of the two parties for only a few minutes at a time. Our result is the following:

\begin{theorem} \label{thm:qges}
The class of functions that a client can delegate to a server in an offline QGES is contained in $\qgesdec$.
\end{theorem}
\noindent Note that the class $\qgesdec$ can be seen as a quantum analogue of the class $\gesdec$ which we encounter in the GES case. We therefore view Theorem~\ref{thm:qges} as a ``quantisation'' of the Abadi et al.\ bound on the power of generalised encryption schemes.

Again, in the spirit of the Abadi et al.\ result, one can ask whether $\cNP$-complete functions are contained in $\qgesdec$. In other words: does giving quantum capabilities to the client increase the class of functions that it can securely delegate so that this class contains $\cNP$?
We give an indication that the answer is no:
\begin{theorem} \label{thm:QGESNP}
$\cNP \subset \qgesdec$ implies $\mathsf{coNP^{NP^{NP}}} \subseteq \cNP^{\cNP^{\cPromiseQMA}}$.
\end{theorem}
Note that if $\mathsf{PromiseQMA}$ in the above expression were replaced with $\cNP$, this would imply a collapse of the polynomial hierarchy at the third level. Our result is as close to a collapse of the polynomial hierarchy as one
can reasonably hope to get, given a quantum hypothesis.
Hence, while a QGES does allow the client to delegate $\cBQP$ computations, it seems to be no more useful than the regular GES for delegating $\cNP$-hard functions.

One could ask why we would even be interested in delegating $\cNP$-hard problems to a quantum computer, given that we do not expect quantum computers to be able to solve such problems in polynomial time \cite{aaronson2004limits}. 
First of all, from a theoretical perspective, note that in the QGES formalism we are not limiting the server to polynomial-time quantum computations, but instead assuming that it has unbounded computational power. Therefore, the way to view this result is not as ``how can a client blindly delegate the evaluation of $\cNP$-hard functions to a quantum computer?'' but as ``can quantum communication help in blindly delegating the evaluation of $\cNP$-hard functions to an unbounded server?''.

From a practical perspective, while we do not believe that quantum computers can solve $\cNP$-complete problems in polynomial time, they could, in principle, solve such problems quadratically faster than classical computers, thanks to Grover's algorithm \cite{grover}. 
Even though the speedup of Grover's algorithm is only
quadratic, from (say) $2^n$ to $2^{n/2}$, our result is only concerned with
the length of the computation performed on the client side, and
therefore applies to Grover's algorithm just as it would to a quantum
algorithm achieving exponential speedup.
In fact, as is mentioned in \cite{pnp}, there are $\cNP$-complete problems for which quantum
computers provide a superpolynomial speedup, at least with respect to the best
known classical algorithms. Our no-go theorem indicates that clients
cannot exploit such speedups by delegating the computation to the server, even
when allowing some quantum communication, if we also want to keep their inputs
hidden in an information-theoretic sense.

Proofs of these results can be found in Section~\ref{sect:qges}.

\subsection{Related work}
As mentioned, the problem of computing on encrypted data was first considered by Rivest, Adleman and Dertouzos \cite{rad}, which then led to the development of \emph{homomorphic encryption} and eventually to fully homomorphic encryption with Gentry's scheme \cite{gentry}. Since then there have been many other FHE protocols relying on more standard cryptographic assumptions and having more practical requirements \cite{fhe1, fhe2, fhe3}.

While FHE is similar to the GES in many respects, there are also significant differences.
For starters, FHE protocols have only one round of interaction between the client and the server, whereas a GES allows for polynomially many rounds.
Additionally, the GES assumes the server is computationally unbounded and hence requires information-theoretic security. In contrast, FHE relies on computational security. More precisely FHE schemes have semantic security against polynomial-time (quantum) algorithms \cite{gentry}.

The problem of \emph{quantum} computing on encrypted data was introduced by Childs \cite{Childs:2005:SAQ:2011670.2011674} and Arrighi and Salvail \cite{doi:10.1142/S0219749906002171}. Further development eventually led to UBQC \cite{bfk,abe} and a scheme of Broadbent \cite{doi:10.1139/cjp-2015-0030}. The latter was followed by the construction of the first schemes for quantum fully homomorphic encryption (QFHE) \cite{broadbent, dulek}.  For a review of blind quantum computing and related protocols see \cite{fitzsimons2016private}.

In the QFHE schemes of \cite{broadbent, dulek}, the server is a polynomial-time quantum computer and the client has some quantum capabilities of its own, although it is not able to perform universal quantum computations. Both the size of the exchanged messages and the number of operations of the client are polynomial in the size of the input. More recently, QFHE schemes have been proposed in which the client is completely classical \cite{mahadev2017classical, cryptoeprint:2018:338}. Similar to FHE, these protocols rely on computational assumptions for security \cite{Alagic2016} and involve one round of back and forth interaction between the client and the server.
QFHE with information-theoretic security (and a computationally unbounded server) has been considered by Yu et al.\ in \cite{qfheimposs}, where it is shown that it is impossible to have such a scheme for arbitrary unitary operations (or even arbitrary reversible classical computations). This result was later reproved by Newmann and Shi using quantum random-access codes \cite{newmannshi}.
In relation to our work, QFHE with information-theoretic security can be viewed as a one-round QGES in which the server responds with a quantum message. The complexity-theoretic upper bound we prove for QGES computable functions would then apply to QFHE as well (provided that in QFHE we only leak the size of the input to the server), since our proof allows a quantum message from the server just as well as a classical message.

The possibility of a classical client delegating a blind computation to a quantum server was considered by Morimae and Koshiba \cite{Morimae}. They showed that such a protocol in which the client leaks no information about its input to the server and there is only one round of interaction leads to $\cBQP \subseteq \cNP$, considered an unlikely containment.
We consider the more general setting of a GES for $\cBQP$ functions, where the number of rounds can be polynomial in the size of the input and we allow the encryption to leak the size of the input.
In fact, the question of whether a GES, as defined in Abadi et al. \cite{afk}, can exist for quantum computations was raised before by Dunjko and Kashefi \cite{vedran}.

\subsection{Future work}
As we remarked in Section~\ref{subsect:main}, in the case of decision problems, the existence of a GES with bounded communication, for polynomial-time quantum computations, leads to the inclusion $\cBQP \subset \cNPnd$. We argue that this containment is unlikely based on the existence of an oracle separating the two complexity classes. A natural extension of this result would be to prove an oracle separation between $\cBQP$ and $\cNPpoly$. This would provide more compelling evidence that a GES for quantum computations cannot exist.

In the case of sampling problems, we showed that a GES for $\textsc{BosonSampling}$ implies the existence of circuits of size $2^{n - \mathsf{\Omega} \left(\frac{n}{\log n} \right)}$, making polynomially-sized queries to an $\cNP^\cNP$ oracle, for computing matrix permanents. Can this result be strengthened so as to provide circuits for computing matrix permanents that would be ruled out by the \emph{strong exponential-time hypothesis}? Alternatively, could one use other quantum sampling problems (such as random circuit sampling or $\sf{IQP}$ problems \cite{harrow2017quantum}) to show that having a GES for such a problem leads to a collapse of the polynomial hierarchy?

We also defined the QGES, which extends the GES by allowing the client to send one quantum message to the server, and gave an upper bound for the set of functions that can be delegated using an offline version of such a scheme. The immediate question one could ask is: what upper bound can we give for an online QGES? 
A related question is: what upper bound can we give for a QGES that allows all of the communication between the client and the server to be quantum? The difficulty in answering both of these questions is that the offline property of the QGES is what allowed us to relate the set of functions that can be delegated to advice classes. Without this property, it seems that a different approach would be needed to provide a complexity-theoretic upper bound.

Another direction that can be explored has to do with the size of the quantum communication between the client and the server. In a QGES in which the client's quantum message is logarithmic or poly-logarithmic in the size of the input (while the classical communication is still polynomial), is it still possible to delegate $\cBQP$ functions to the server?  Of course, this question only makes
sense if we assume that the client is not able to perform $\cBQP$ computations itself.

\section{Preliminaries} \label{sect:background}
\subsection{Quantum information and computation basics} \label{sect:qinfo}
In this subsection we provide a few basic notions regarding quantum information and quantum computation and refer the reader to the appropriate references for a more in depth presentation \cite{nc, watrous2009quantum}.

A quantum state (or a quantum register) is a unit vector in a complex Hilbert space, $\mathcal{H}$. We denote quantum states, using standard Dirac notation, as $\Ket{\psi} \in \mathcal{H}$, called a `ket' state. The dual of this state is denoted $\Bra{\psi}$, called a `bra', and is a member of the dual space $\mathcal{H}^{\perp}$.
We will only be concerned with finite-dimensional Hilbert spaces.
Qubits are states in two-dimensional Hilbert spaces. Traditionally, one fixes an orthonormal basis for such a space, called \emph{computational basis}, and denotes the basis vectors as $\Ket{0}$ and $\Ket{1}$. Gluing together systems to express the states of multiple qubits is achieved through \emph{tensor product}, denoted $\otimes$. The notation $\Ket{\psi}^{\otimes n}$ denotes a state comprising of $n$ copies of $\Ket{\psi}$. If a state $\Ket{\psi} \in \mathcal{H}_1 \otimes \mathcal{H}_2$ cannot be expressed as $\ket{a} \otimes \ket{b}$, for any $\ket{a} \in \mathcal{H}_1$ and any $\ket{b} \in \mathcal{H}_2$, we say that the state is \emph{entangled}.

Quantum mechanics dictates that there are two ways to change a quantum state: \emph{unitary evolution} and \emph{measurement}. Unitary evolution involves acting with some unitary operation $U$ (so $U U^{\dagger} = U^{\dagger} U = I$, where the $\dagger$ operation denotes the hermitian adjoint, obtained through transposing and complex conjugating) on $\Ket{\psi}$, thus producing the mapping $\Ket{\psi} \rightarrow U \Ket{\psi}$.

Measurement, in its most basic form, involves expressing a state $\ket{\psi}$ in a particular orthonormal basis, $\mathcal{B}$, and then choosing one of the basis vectors as the state of the system post-measurement. The index of that vector is the classical outcome of the measurement. The post-measurement vector is chosen at random and the probability of obtaining a vector $\ket{v} \in \mathcal{B}$ is given by $| \braket{v|\psi} |^2$. There are more general types of measurement, however this is the only type that is relevant to our paper.

States denoted by kets are also referred to as \emph{pure states} as they are states of maximal information for a quantum system. In other words, having a pure state for a particular quantum system means knowing all there is to know about the state of that system. When maximal information is not available, states are referred to as \emph{mixed} and can be represented using \emph{density matrices}. These are positive semidefinite, trace one, hermitian operators. The density matrix of a pure state $\ket{\psi}$ is $\rho = \Ket{\psi} \bra{\psi}$.

An essential operation concerning density matrices is the \emph{partial trace}. This provides a way of obtaining the density matrix of a subsystem that is part of a larger system. Partial trace is linear, and is defined as follows. Given two density matrices $\rho_1$ and $\rho_2$ with Hilbert spaces $\mathcal{H}_1$ and $\mathcal{H}_2$, we have that:
\begin{equation}
\rho_1 = Tr_2(\rho_1 \otimes \rho_2) \; \; \; \; \; \; \; \rho_2 = Tr_1(\rho_1 \otimes \rho_2)
\end{equation}
In the first case one is `tracing out' system $2$, whereas in the second case we trace out system $1$. This property together with linearity completely defines the partial trace. For if we take any general density matrix, $\rho$, on $\mathcal{H}_1 \otimes \mathcal{H}_2$, expressed as:
\begin{equation}
\rho = \sum_{i,i',j, j'} a_{ii'jj'} \Ket{i}_1\Bra{i'}_1 \otimes \Ket{j}_2 \Bra{j'}_2
\end{equation}
where $\{\ket{i}\}$ ($\{\ket{i'}\}$) and $\{\ket{j}\}$ ($\{\ket{j'}\}$) are orthonormal bases for $\mathcal{H}_1$ and $\mathcal{H}_2$,
if we would like to trace out subsystem $2$, for example, we would then have:
\begin{equation}
Tr_2(\rho) = Tr_2 \left( \sum_{i,i',j, j'} a_{ii'jj'} \Ket{i}_1\Bra{i'}_1 \otimes \Ket{j}_2 \Bra{j'}_2 \right) =
\sum_{i,i',j} a_{ii'jj} \Ket{i}_1\Bra{i'}_1
\end{equation}

An important result, concerning the relationship between mixed states and pure states which we use in our paper, is the fact that any mixed state can be purified. In other words, for any mixed state $\rho$ over some Hilbert space $\mathcal{H}_1$ one can always find a pure state $\ket{\psi} \in \mathcal{H}_1 \otimes \mathcal{H}_2$ such that $dim(\mathcal{H}_1) = dim(\mathcal{H}_2)$\footnote{One could allow for purifications in larger systems, but in our paper we restrict attention to same dimensions.} and:
\begin{equation}
Tr_2(\Ket{\psi}\Bra{\psi}) = \rho
\end{equation}
Moreover, the purification $\ket{\psi}$ is not unique and so another important result is the fact that if $\ket{\phi} \in \mathcal{H}_1 \otimes \mathcal{H}_2$ is another purification of $\rho$ then there exists a unitary $U$, acting only on $\mathcal{H}_2$ (the additional system that was added to purify $\rho$) such that:
\begin{equation}
\ket{\phi} = (I \otimes U) \ket{\psi}
\end{equation}
We will refer to this as the \emph{purification principle}.

Quantum computation is most easily expressed in the \emph{quantum gates model}. In this framework, gates are unitary operations which act on groups of qubits. As with classical computation, universal quantum computation is achieved by considering a fixed set of quantum gates which can approximate any unitary operation up to a chosen precision. The most common universal set of gates is given by:
\begin{equation*}
\mathsf{X} = \begin{bmatrix}
0 & 1 \\
1 & 0 \\
\end{bmatrix} \;
\mathsf{Z} = \begin{bmatrix}
1 & 0 \\
0 & -1 \\
\end{bmatrix} \;
\mathsf{H} = \frac{1}{\sqrt{2}} \begin{bmatrix}
1 & 1 \\
1 & -1 \\
\end{bmatrix} \;
\mathsf{T} = \begin{bmatrix}
1 & 0 \\
0 & e^{i\pi/4} \\
\end{bmatrix} \;
\mathsf{CNOT} = \begin{bmatrix}
		1 & 0 & 0 & 0 \\
		0 & 1 & 0 & 0 \\
		0 & 0 & 0 & 1 \\
		0 & 0 & 1 & 0 \\		
		\end{bmatrix}
\end{equation*}
In order, the operations are known as Pauli $\mathsf{X}$ and Pauli $\mathsf{Z}$, Hadamard, the $\mathsf{T}$-gate and controlled-NOT. Note that general controlled-$U$ operations are operations performing the mapping $\Ket{0}\ket{\psi} \rightarrow \ket{0}\ket{\psi}$, $\ket{1}\ket{\psi} \rightarrow \ket{1} U\ket{\psi}$. The matrices express the action of each operator on the computational basis.
A classical outcome for a particular quantum computation can be obtained by measuring the quantum state resulting from the application of a sequence of quantum gates.

The final notion which needs mentioning is the \emph{quantum} \textsf{SWAP} \emph{test}. This is a simple procedure for determining whether two quantum states $\ket{\psi}, \ket{\phi} \in \mathcal{H}$ are close to each other or far apart. We express closeness in terms of the absolute value of their inner product $|\braket{\psi | \phi}|$.
The test involves preparing a qubit in the state $(\ket{0} + \ket{1}) / \sqrt{2}$ and performing a controlled-\textsf{SWAP} operation between that qubit and the state $\ket{\psi}\ket{\phi}$. \textsf{SWAP} is defined by the mapping $\ket{\psi}\ket{\phi} \rightarrow \ket{\phi}\ket{\psi}$, so we obtain the state:
\[
\frac{\ket{0}\ket{\psi}\ket{\phi} + \ket{1}\ket{\phi}\ket{\psi}}{\sqrt{2}}
\]
If one then applies a Hadamard operation to the first qubit and measures it in the computational basis it can be shown that the probability of obtaining outcome $\ket{0}$ is $(1 + |\braket{\psi | \phi}|^2)/2$.

\subsection{Complexity theory} \label{sect:complexity}

\subsubsection{Decision problems}
We use standard complexity theory notation and refer the reader to the Complexity Zoo \cite{zoo} for the definitions of standard complexity classes. Briefly, $\cP$ is the class of decision problems\footnote{A decision problem is a problem in which for every input $x \in \{0, 1\}^*$, the output is either ``yes'' or ''no''. A decision problem can be represented as either a function $f:\{0, 1\}^* \rightarrow \{0, 1\}$, or as a subset of $\{0, 1\}^*$ representing the ``yes'' instances to the problem. Such a set is known as a \emph{language}.} that can be solved by a deterministic polynomial-time classical algorithm (or Turing machine). If the algorithm is allowed to use randomness (and we require that it outputs the correct answer with probability greater than $2/3$) the corresponding class is $\cBPP$. The class of decision problems for which ``yes'' instances admit a polynomial-sized proof string (or witness) that can be verified by the polynomial-time algorithm is known as $\cNP$\footnote{There is an equivalent definition of $\cNP$ as the set of all decision problems that can be solved by a \emph{non-deterministic} polynomial-time algorithm (or Turing machine). The non-deterministic part simply means that the algorithm can guess a witness.}. The analogous class for ``no'' instances is $\ccoNP$, referred to as the \emph{complement} of $\cNP$. Once again, if the polynomial-time algorithm also uses randomness the corresponding classes are $\cMA$ and $\mathsf{coMA}$, respectively.
The class of decision problems that can be decided in polynomial-time by a quantum algorithm is denoted $\cBQP$. There are two quantum analogues of $\cMA$. One is $\cQMA$ which is the class of decision problems in which ``yes'' instances admit a quantum witness, of polynomially-many qubits, that can be verified by a polynomial-time quantum algorithm. The other is $\cQCMA$, which is the same as $\cQMA$ except the witness is a classical bit-string, rather than a quantum state. Complements of these classes are denoted $\mathsf{coQMA}$ and $\mathsf{coQCMA}$, respectively.
For all classes mentioned here, we say that a problem $P$ is \emph{hard} for the class $\mathcal{C}$ if for all problems $P' \in \mathcal{C}$ there exists a deterministic polynomial-time algorithm mapping each ``yes'' instance of $P'$ to a ``yes'' instance of $P$ and each ``no'' instance of $P'$ to a ``no'' instance of $P$, respectively. We also say that $P$ is \emph{complete} for $\mathcal{C}$ if $P$ is hard for $\mathcal{C}$ and $P \in \mathcal{C}$.

As a slight abuse of terminology, we will sometimes say ``$\cBPP$ machine/algorithm'' or ``$\cBQP$ machine/algorithm'' to mean either a probabilistic polynomial-time algorithm, or a polynomial-time quantum algorithm, respectively. Analogous terminology may be used for the other classes as well\footnote{For instance, a ``$\mathsf{QCMA}$ machine'' refers to a polynomial-time quantum algorithm that also receives a classical witness string.}.

An important category of complexity classes that we encounter throughout the paper is that of \emph{advice classes}. Let us provide a definition of this concept:
\begin{definition} \label{defn:adviceclass}
Let $\mathcal{C}$ be a complexity class and $\mathcal{F}$ a family of functions $f : \mathbb{N} \rightarrow \{0, 1\}^*$. The complexity class $\mathcal{C}/\mathcal{F}$, known as $\mathcal{C}$ with $\mathcal{F}$ advice, is the set of all languages $L$, for which there exists an $L' \in \mathcal{C}$ and a function $f \in \mathcal{F}$ such that for all $x \in \{0, 1\}^*$, $x \in L$ iff. $\braket{x, f(|x|)} \in L'$.
\end{definition}
\noindent As an example, consider the class $\cPpoly$. This consists of all languages that can be decided by a deterministic polynomial-time algorithm, that receives polynomially-many bits of advice for all inputs of the same length. In other words, for all inputs $x \in \{0, 1\}^n$, the algorithm also receives some string $a \in \{0, 1\}^{poly(n)}$, aiding it in deciding whether to accept $x$ or not.
Analogously $\cNPpoly$ consists of languages that can be decided by a polynomial-time verifier receiving a witness for ``yes'' instances and a trusted advice string that only depends on the size of the input.
We also encounter the class $\cNPnd$ in which the size of the advice string is $\mathsf{O(n^d)}$, for some fixed constant $d$.

Some of the advice classes used in the paper are not covered by Definition~\ref{defn:adviceclass}. For instance, the class $\sf{BPP/rpoly}$ denotes the set of languages that can be decided by a $\cBPP$ machine that receives \emph{randomised} polynomial-size advice. In other words, for all inputs $x \in \{0, 1\}^n$, the probabilistic algorithm also receives some string $a \in \{0, 1\}^{poly(n)}$, that is drawn from a distribution $\mathcal{D}_n$.
We can see that this does not satisfy Definition~\ref{defn:adviceclass} since the advice string is not the result of some deterministic function, but is a sample from a probability distribution.
It should therefore be understood that $\mathsf{rpoly}$ corresponds to polynomial-size advice drawn from a probability distribution that only depends on the size of the input. An important result we use, concerning randomised advice, is the following:
\begin{theorem}[Aaronson \cite{marpoly}] \label{thm:MANP}
$\cMArpoly = \cMApoly = \cNPpoly$
\end{theorem}
\noindent It should be noted that the equalities are only known to hold for classes of decision problems.

We can similarly have quantum advice. As an example, the class $\sf{BQP/qpoly}$ denotes the set of languages that can be decided by a $\cBQP$ machine that receives as advice a quantum state of polynomially-many qubits. In other words, for all inputs $x \in \{0, 1\}^n$, the quantum algorithm also receives a quantum state $\ket{\psi_n} \in \mathcal{H}_n$, such that $dim(\mathcal{H}_n) = 2^{poly(n)}$.
Hence, the suffix $\mathsf{qpoly}$ will indicate polynomially-sized quantum advice and represents a quantum state of polynomially-many qubits that only depends on the size of the input.

The concept of oracles is also used throughout the paper.
Briefly, an oracle is a black box function that can be invoked by an algorithm (either classical or quantum) in order to obtain the solution to some problem in one time step. For example, the class of problems which can be solved by a deterministic polynomial-time algorithm with access to some oracle function $O : \{0,1\}^* \rightarrow \{0,1\}$ is denoted $\cP^O$. If $O$ is an oracle for some $\cNP$-complete problem, then the corresponding class is denoted $\cP^{\cNP}$.
Oracles can be used to define the polynomial hierarchy. The zeroth level of the polynomial hierarchy is given by the classes $\Sigma_0^{\sf P} = \cP$, $\Pi_0^{\sf P} = \cP$. The $k$'th level of the hierarchy is then defined as $\Sigma_k^{\sf P} = \cNP^{\Sigma_{k-1}^{\sf P}}$, $\Pi_k^{\sf P} = \ccoNP^{\Sigma_{k-1}^{\sf P}}$. Finally, the polynomial hierarchy is defined as $\cPH = \cup_{k \geq 0} \Sigma_k^{\sf P}$. We say that the polynomial hierarchy collapses at level $k$ iff. $\Sigma_k^{\sf P} = \Pi_k^{\sf P}$. While not a decision class, we also mention $\csP$ which is the class of all functions $f : \{0, 1\}^* \rightarrow \mathbb{N}$ that take as input a description of a polynomial-time algorithm and output the number of inputs that the algorithm accepts (i.e. the number of ``yes'' instances). An important result in complexity theory is Toda's theorem \cite{toda}, which states that $\cPH \subseteq \cP^{\csP}$.

For quantum classes, oracles are viewed as unitary operations that perform mappings of the form $U_O \ket{x}\ket{y} = \ket{x}\ket{y \oplus O(x)}$, where $O$ is the oracle function.
Additionally, whenever a result involving complexity classes remains true when those classes are given access to an oracle, $O$, we say that the result \emph{relativises}.

\subsubsection{Sampling problems and \textsc{BosonSampling}} \label{prelim:samp}
This section discusses sampling problems. These are problems for which the input specifies a probability distribution and the goal is to sample either exactly or approximately from that distribution.
In this paper, we will only be interested in exact sampling and specifically in the $\textsc{BosonSampling}$ problem, defined by Aaronson and Arkhipov \cite{bosonsampling}.
As mentioned in the introduction, in \textsc{BosonSampling}, identical photons (bosons) are sent through a linear optics network and non-adaptive measurements are performed to count the number of photons in each mode.
In more detail, for a quantum system with $m$ modes and $n$ photons, the basis states of the system are of the form $S = (s_1, ... s_m)$, where $s_i$ denotes the number of photons in mode $i$ (so $s_1 + ... + s_m = n$). A general state, is then a state of the form:
\begin{equation}
\ket{\psi} = \sum_{S} \alpha_S \ket{S}, \text{ with } \sum_S |\alpha_S|^2 = 1
\end{equation}
Note that the number of basis states is $M = {m + n - 1 \choose n}$. The action of the linear optics network can be expressed as a matrix $A \in \mathcal{U}_{m,n}$, where $\mathcal{U}_{m,n}$ is the set of all $m \times n$ column-orthonormal matrices. Let $A_{S}$ be the matrix obtained by taking $s_i$ copies of the $i$'th row of $A$, for all $i \leq m$. If the initial state of the system consists of one photon in each of the first $n$ modes (it is assumed that $m \geq n$), a state denoted as $\ket{1_n} = \ket{1,...,1,0,...,0}$, then it can be shown (see \cite{bosonsampling}) that the probability of observing the state $S$, upon passing the photons through the network described by $A$ and measuring the number of photons in each mode, will be:
\begin{equation} \label{eqn:bssamp}
Pr(S) = \frac{|Per(A_S)|^2}{s_1! s_2! ... s_m!}
\end{equation}
where $Per(M)$ denotes the permanent of a matrix $M = (m_{ij})_{i,j \leq n}$, and is defined as:
\begin{equation}
Per(M) = \sum\limits_{\sigma \in S_n} \prod\limits_{i=1}^n m_{i, \sigma(i)}
\end{equation}
with $S_n$ the symmetric group of all permutations of the elements $1$ up to $n$.

Exact \textsc{BosonSampling} is then the problem of sampling from the distribution defined by Equation~\ref{eqn:bssamp}.
This problem is believed to be hard for classical computers and to explain why we first need to state a result known as \emph{Stockmeyer's approximate counting method}:
\begin{theorem}[Stockmeyer \cite{stockmeyer}]
Let $f : \{0, 1\}^n \rightarrow \{0, 1\}$ be a function that can be computed by a deterministic polynomial-time algorithm and let:
\begin{equation} \label{eqn:stockmeyer}
p = \frac{1}{2^n}\sum_{x \in \{0, 1\}^n} f(x)
\end{equation}
Then for all $g \geq 1 + 1/poly(n)$ there exists a $\cBPP^{\cNP}$ algorithm that computes $p$ to within a multiplicative factor of $g$\footnote{In other words, the algorithm computes $\tilde{p}$ such that $p/g \leq \tilde{p} \leq gp$.}.
\end{theorem}
Now, suppose there existed a $\cBPP$ algorithm that, given $A$ (the description of the linear optics network) as input, could sample from the distribution of Equation~\ref{eqn:bssamp}. 
This algorithm can be viewed as a deterministic polynomial-time computable function $F$ that, given $A$ and a string $r \in \{0,1\}^{p(n)}$, drawn from the uniform distribution, for some polynomial $p$, produces a vector $S = (s_1, ... s_m)$ (of the form described above).
The fact that this algorithm can sample from the \textsc{BosonSampling} distribution can be expressed mathematically as:
\begin{equation}
\underset{r \leftarrow_R \{0, 1 \}^{p(n)}}{Pr} \left( F(A, r) = S \right) = \frac{|Per(A_S)|^2}{s_1! s_2! ... s_m!}
\end{equation}
where $r \leftarrow_R \{0, 1 \}^{p(n)}$ denotes the fact that $r$ was drawn uniformly at random from the set $\{0, 1 \}^{p(n)}$.

Consider now the state $\ket{1_n}$ (or any state in which all $s_i$ are either $0$ or $1$). What is the probability of observing $\ket{1_n}$ in the output modes? Using Equation~\ref{eqn:bssamp} we see that it is $Pr(\ket{1_n}) = |Per(A_{\ket{1_n}})|^2$. We will define a function $f$ as follows:
\begin{equation}
f(A, r) = \left\{
\begin{array}{ll}
0, & \text{ if } F(A, r) \neq \ket{1_n} \\
1, & \text{ if } F(A, r) = \ket{1_n} \\
\end{array}
\right.
\end{equation}
Note that $f$ is computable in polynomial time (since it simply involves evaluating $F$ and testing whether the output is $\ket{1_n}$).
The probability that the $\cBPP$ algorithm produces the output $\ket{1_n}$ can then be expressed as:
\begin{equation}
\underset{r \leftarrow_R \{0, 1 \}^{p(n)}}{Pr} \left( F(A, r) = \ket{1_n} \right) = \frac{1}{2^{p(n)}} \sum_{r \in \{0, 1\}^{p(n)}}  f(A, r)
\end{equation}
But this sum can be estimated, up to multiplicative error, in $\cBPP^{\cNP}$ using Stockmeyer's method. In other words, there is a $\cBPP^{\cNP}$ algorithm for estimating $|Per(A_{\ket{1_n}})|^2$.
It is shown in \cite{bosonsampling} that one can consider any matrix $M \in \{-1, 0, 1\}^{n \times n}$ and embed it in $A$ (with only an added polynomial overhead) so that the probability of sampling the $\ket{1_n}$ state is proportional to $|Per(M)|^2$.
By the above argument, this means that computing a multiplicative estimate for the squared permanent of a matrix over $\{-1, 0, 1\}$ is in $\cBPP^{\cNP}$.
However, computing such an estimate is $\csP$-hard \cite{bosonsampling}. It is also known that $\cBPP^{\cNP}$ is contained in the third level of the polynomial hierarchy \cite{bppph}, which leads us to conclude, using Toda's theorem, that the polynomial hierarchy collapses at the third level.
Such a collapse is regarded as unlikely and therefore the existence of an efficient classical algorithm for $\textsc{BosonSampling}$ is also considered unlikely.

\subsection{Generalised Encryption Scheme (GES)} \label{sect:crypto}
The basis of most of the results in our paper is the generalised encryption scheme. We state its definition from \cite{afk}:
\begin{definition}
[\cite{afk} Generalised Encryption Scheme (GES)] \label{defn:GES}
A generalised encryption scheme (GES) is a two party protocol between a classical client $C$, and an unbounded server $S$, characterised by:
\begin{itemize}
\item A function $f: \{0, 1\}^* \rightarrow \{0, 1\}$.
\item A cleartext input $x \in \{0, 1\}^*$, for which the client wants to compute $f(x)$.
\item An expected polynomial-time key generation algorithm $K$ which works as follows: for any $x \in \{0, 1\}^*$, with probability greater than $1/2 + 1/poly(|x|)$ we have $(k, success) \leftarrow K(x)$, where $k \in \{0,1\}^{poly(|x|)}$. If the algorithm does not return $success$ then we have $(k', fail) \leftarrow K(x)$, where $k' \in \{0,1\}^{poly(|x|)}$.
\item A polynomial-time deterministic algorithm $E$ which works as follows: for any $x \in \{0, 1\}^*$, $k \in \{0, 1\}^{poly(|x|)}$ and $s \in \{0,1\}^{poly(|x|)}$ we have that $y \leftarrow E(x, k, s)$, where $y \in \{0,1\}^{poly(|x|)}$.
\item  A polynomial-time deterministic decryption algorithm $D$, which works as follows: for any $x \in \{0, 1\}^*$, $k \in \{0, 1\}^{poly(|x|)}$ and $s \in \{0,1\}^{poly(|x|)}$ we have that $z \leftarrow D(s, k, x)$, where $z \in \{0,1\}^{poly(|x|)}$.
\end{itemize}
And satisfying the following properties:
\begin{enumerate}
\item There are $m$ rounds of communication, such that $m = poly(|x|)$. Denote the client's message in round $i$ as $c_i$ and the server's message as $s_i$.
\item On cleartext input $x$, $C$ runs the key generation algorithm until success to compute a key $(k, success) = K(x)$. This happens before the communication between $C$ and $S$ is initiated, and the key $k$ is used throughout the protocol.
\item In round $i$ of the protocol, $C$ computes $c_i = E(x, k, \overline{s}_{i-1})$, where $\overline{s}_{i-1}$ denotes the server's responses up to and including round $i-1$, i.e.\ $\langle s_0, s_1 ... s_{i-1} \rangle$. We assume that $s_0$ is the empty string. $C$ then sends $c_i$ to $S$.
\item In round $i$ of the protocol, $S$ responds with $s_i$, such that $s_i \in \{0,1 \}^{poly(|x|)}$. Additionally, the server's responses are drawn probabilistically from a distribution which is consistent with property $5$.
\item At the end of the protocol, $C$ computes $z \leftarrow D(\overline{s}_m, k, x)$ and with probability $1/2 + 1/poly(|x|)$, we have that $z = f(x)$.
\end{enumerate}
\end{definition}
Let us provide some intuition for this definition. The purpose of a GES is to allow a client to compute some $f(x)$ which it cannot compute with its own resources. It does this by interacting with a computationally powerful server for a number of rounds which is polynomial in the size of the input. Importantly, the GES allows the client to hide some information about $x$ from the server. We make this statement more precise through the following definition:
\begin{definition}
Let $X$ be a random variable denoting the input to a GES and $T(X)$ a random variable denoting the transcript of the protocol for input $X$ (in other words $T(X)$ is a collection of all the messages exchanged between the client and the server, in a run of the GES, on input $X$).
We say that a GES leaks at most $L(X)$ iff. $X$ and $T(X)$ are independent given $L(X)$.
\end{definition}
\noindent Finally, we state and give a simplified proof of the main theorem from \cite{afk}, which we will use throughout the paper:
\begin{theorem}[\cite{afk} GES leaking size of input] \label{thm:afk}
If a function $f$ admits a GES which leaks at most the size of the input (i.e.\ $L(X) = |X|$), then $f \in \gesdec$.
\end{theorem}
\begin{proof}
Suppose that $f$ admits a GES which leaks at most the size of the input. We start by first considering the simplified case of a GES with only one round of interaction between the client and the server. The protocol works as follows:
\begin{enumerate}
\item The client runs $K(x)$ until success to produce an encryption key $k$.
\item The client computes the encrypted string $y \leftarrow E(x, k, \text{`'})$ (where the last entry is the empty string) and sends it to the server.
\item The server sends a response $r$.
\item The client decrypts his response obtaining $z \leftarrow D(r, k, x)$. With probability greater than $1/2 + 1/poly(|x|)$ we have that $z = f(x)$.
\end{enumerate}

Assuming the existence of the one-round GES, let us construct an $\cMArpoly$ algorithm for computing $f(x)$. In other words we are going to construct a probabilistic polynomial-time algorithm that receives a checkable witness and randomised polynomial-sized advice (the distribution from which we sample the advice will be the same for inputs of the same size).
The algorithm takes $x$ as input and works as follows:
\begin{itemize}
\item Denoting $|x| = n$, the algorithm receives as advice a string $x_n \in \{0, 1\}^n$ as well as $r_n \in \{0, 1\}^{poly(n)}$, where $r_n$ is the server's response, in the one-round GES, when being sent $y_n \leftarrow E(x_n, k_n, \text{`'})$ from the client. Here $k_n$ is simply some key which can be used to encrypt $x_n$. The only reason we include $x_n$ as part of the advice is so that we can whether $x_n = x$. If this is the case then the algorithm simply decrypts $r_n$ obtaining $f(x)$ with high probability. The next steps assume that $x_n \neq x$.
\item From the assumption that the GES leaks at most the size of the input, there must be some key $k$, such that $y_n \leftarrow E(x, k, \text{`'})$. If there did not exist such a key and the server received $y_n$ he would know that the input could not be $x$ and thus learn more than the size of the input, which is not allowed. 
More formally, it would mean that the input and the transcript of the protocol are not independent, given the length of the input, since certain transcripts (certain $y$'s) can only occur for certain inputs.
The key, $k$, will be the witness that the algorithm receives. The algorithm can check whether $y_n \leftarrow E(x, k, \text{`'})$.
\item The algorithm now simply computes $z \leftarrow D(r_n, k, x)$, which by the definition of the GES, will be $f(x)$ with probability greater than $1/2 + 1/poly(n)$.
\end{itemize}
We have therefore given an $\cMArpoly$ algorithm for computing $f(x)$. To recap, the $\cMA$ part comes from the fact that the algorithm is probabilistic, runs in polynomial time and receives the key $k$ as a witness. The advice is $\sf{rpoly}$ because the server's response is drawn from some probability distribution (which depends only on the length of the input).
From Theorem~\ref{thm:MANP}, we know that $\cMArpoly = \cNPpoly$ hence $f \in \cNPpoly$. 
Since the GES frameworks requires that the key $k$ must exist irrespective of the value of $f(x)$, this means that in our algorithm both the $f(x) = 1$ case and the $f(x) = 0$ have a verifiable witness. Therefore it is also the case that $f \in \ccoNPpoly$ and so $f \in \gesdec$.

We now need to generalise this to the case where the client and the server interact for a polynomial number of rounds. Because the protocol is leaking at most the size of the input, denoted $n$, any transcript of the protocol will only depend on $n$. Therefore we can make the algorithm's advice to be a complete transcript of the protocol drawn from the distribution of all possible transcripts for inputs of length $n$. The witness would then be a key $k$ that makes the input $x$ compatible with this transcript. From the definition of the GES this again guarantees that we obtain the correct outcome with probability $1/2 + 1/poly(|x|)$.

Note that if the total communication between the client and the server (i.e. the size of the transcript) were bounded by $\mathsf{O(n^d)}$, for some constant $d > 0$, the above argument shows that the functions computable in this setting are contained in $\cNPnd$\footnote{Strictly speaking, the above argument shows that such functions would be contained in $\mathsf{MA}$ with \emph{randomised} advice of size $\mathsf{O(n^d)}$. However, the proof that $\mathsf{MA/rpoly} = \cNPpoly$ can be adapted to show that $\cMA$ with $\mathsf{O(n^d)}$-size randomised advice is the same as $\mathsf{MA/O(n^{d'})}$, with $d' \geq d$. Essentially, we can ``derandomise'' the advice using deterministic advice of a larger size. The same argument cannot be used, however, to derandomise the $\cMA$ part of the algorithm and obtain $\mathsf{NP/O(n^d)}$. This is because the size of the randomness used in $\cMA$ is an arbitrary polynomial.}. This is because, as we have seen, the transcript is given as advice and so it will also be bounded in length by $\mathsf{O(n^d)}$. 
\end{proof}

We have seen that the functions that can be delegated in a GES are contained in $\gesdec$. A question we might have regarding this result is: what can we say about $\cBPP^{\gesdec}$? In other words, if a $\cBPP$ machine uses the GES as an oracle, does that allow it to solve more problems? Intuitively, we would expect the answer to be no and indeed using a result of Brassard \cite{brassard} which shows that $\cP^{\cNP \cap \ccoNP} = \cNP \cap \ccoNP$, together with Adleman's theorem (that $\cBPP \subset \cPpoly$) \cite{adleman} we prove the following:
\begin{lemma}
$\cBPP^{\gesdec} = \gesdec$.
\end{lemma}
\begin{proof}
It is clear that $\gesdec \subseteq \cBPP^{\gesdec}$, so only we need to show that $\cBPP^{\gesdec} \subseteq \gesdec$. To do this, we first use Adleman's theorem \cite{adleman}, that $\cBPP \subset \cPpoly$, which we know is relativising and have that $\cBPP^{\gesdec} \subseteq \cPpoly^{\gesdec}$. Next, it is easy to show that $\cPpoly^{\gesdec} \subseteq \cP^{\gesdec}$. This is because the advice received by the $\cPpoly$ machine can just as easily be obtained from the $\gesdec$ oracle. In other words, for any given input $x$ and advice $a$ for the $\cPpoly$ machine, the $\cP$ machine can simply query the $\gesdec$ oracle with $x$ in order to obtain the same advice $a$\footnote{The fact that the oracle responds with a single bit (acceptance or rejection) is not a problem, since the $\cP$ machine can query the oracle for each bit of $a$.}. It then simulates the $\cPpoly$ machine.

We have therefore reduced our problem to showing that $\cP^{\gesdec} \subseteq \gesdec$. This can be done by adapting Brassard's proof \cite{brassard} that $\cP^{\cNP \cap \ccoNP} = \cNP \cap \ccoNP$. The essential part of that proof is to show that $\cP^{\cNP \cap \ccoNP} \subseteq \cNP$, while the containment in $\ccoNP$ follows by complementation. The idea is that for any $\cP^{\cNP \cap \ccoNP}$ algorithm, $A$, deciding some language, we can devise an $\cNP$ algorithm, \emph{NA}, which also decides that language.

The \emph{NA} algorithm will simulate $A$ until it makes a query to the $\cNP \cap \ccoNP$ oracle. At this point \emph{NA} can non-deterministically guess the response to this query. To do so, note that if some language $L \in \cNP \cap \ccoNP$ then it is the case that $L \in \cNP$ and $L^c \in \cNP$, where $L^c$ is the complement of $L$. In other words, there exist non-deterministic algorithms $N_L$ and $N_{L^c}$ for deciding $L$ and $L^c$, respectively. Assuming $A$'s query is for the language $L$, \emph{NA} will simulate $N_L$, and for each non-deterministic branch of this simulation it will then also simulate $N_{L^c}$. Since $L$ and $L^c$ are complementary, it cannot happen that both the $N_L$ and the $N_{L^c}$ parts of the branches are accepting. We will therefore have branches in which both $N_L$ and $N_{L^c}$ were rejecting and branches in which either $N_L$ was accepting or $N_{L^c}$ was accepting. These latter branches determine the answer to the query for the $\cNP \cap \ccoNP$ oracle. The \emph{NA} algorithm will continue simulating $A$ on these branches and reject on all others.

We can see that the above reasoning would also work if the oracle was $\gesdec$ and the algorithm \emph{NA} were an $\cNPpoly$ algorithm receiving some advice string whose length is polynomial in the size of the input. Our modified \emph{NA} can continue to simulate the oracle queries if we assume that the advice it receives is the concatenation of advices received by the $\gesdec$ oracle for all queries. Since the number of queries is polynomial, the concatenation will also be polynomially bounded and hence constitutes a valid advice string for an $\cNPpoly$ algorithm. Therefore $\cP^{\gesdec} \subseteq \cNPpoly$ and through complementation $\cP^{\gesdec} \subseteq \gesdec$.

Because $\cBPP^{\gesdec} \subseteq \cP^{\gesdec}$, our result follows immediately.
\end{proof}

We end this section with an explanation of what it means to have a GES for sampling problems. The client's input, $x$, will be a description of a probability distribution, $\mathcal{D}_x : \{0, 1\}^{poly(|x|)} \rightarrow [0, 1]$. Upon interacting with the server and applying the decryption procedure the client obtains $z \leftarrow D(\overline{s}_m, k, x)$ such that:
\begin{equation}
Pr(z \leftarrow D(\overline{s}_m, k, x)) = Pr(z \leftarrow \mathcal{D}_x)
\end{equation}
In other words, the output of $D(\overline{s}_m, k, x))$ is distributed according to $\mathcal{D}_x$. Just as with the GES for decision problems, throughout the interaction, the server should only be able to learn $|x|$.
For the specific case of $\textsc{BosonSampling}$, with $m$ modes and $n$ photons, the input will be an $m \times n$ column-orthonormal matrix $A$. The associated distribution will be the one described in Section~\ref{prelim:samp}, namely:
\begin{equation}
Pr(S) = \frac{|Per(A_S)|^2}{s_1! s_2! ... s_m!}
\end{equation}
where $S$ is a particular configuration of the $n$ photons in the $m$ modes.

The proof of Theorem~\ref{thm:afk} applies to the sampling case as it did to the decision case. We therefore have that if the client can delegate exact sampling from $\mathcal{D}_x$ to the server, using the GES, there exists an $\mathsf{MA/rpoly}$ algorithm for exactly sampling from $\mathcal{D}_x$. Importantly, however, the result of Theorem~\ref{thm:MANP} no longer applies and we cannot equate this algorithm with a $\mathsf{NP/poly}$ sampling algorithm. That result \emph{applies to decision classes}. In fact, in the sampling case, it will be simpler to consider the sampling algorithm as a $\mathsf{{BPP}^{NP}/rpoly}$ algorithm (since $\mathsf{MA/rpoly} \subseteq \mathsf{{BPP}^{NP}/rpoly}$). In other words, the existence of a GES for sampling from $\mathcal{D}_x$ implies the existence of a probabilistic polynomial-time algorithm, with an $\cNP$ oracle, and which receives randomised polynomial-sized advice, for sampling from $\mathcal{D}_x$.

\section{Oracle separation between \textsf{BQP} and \textsf{MA/O(n\textsuperscript{d})}} \label{sect:proofs1}
In order to prove Theorem~\ref{thm:noGESBQP} we will construct an oracle using a version of the complement of Simon's problem \cite{simon}. Recall that Simon's problem is the following: given a function $f : \{0, 1\}^n \rightarrow \{0, 1\}^n$ (for some $n \in \mathbb{N}$) which is promised to be either $1$-to-$1$ or have Simon's property ($f$ is $2$-to-$1$ and there exists some $s \in \{0, 1\}^n$, $s \neq 0^n$, such that for $x \neq y$, $f(x) = f(y)$ iff $x = s \oplus y$), decide which is the case. In particular, for Simon's problem, the deciding algorithm should accept if the function has Simon's property and reject if it is a $1$-to-$1$ function. The complement of this problem simply flips these two conditions.
If one is not given an explicit description of $f$ but restricts access to this function through an oracle then Simon's problem can be used to separate $\cBPP$ from $\cBQP$.
To be precise, the oracle is some function $O: \{0, 1\}^* \rightarrow \{0, 1\}^*$ such that for $n \in \mathbb{N}$, if we consider $O$ restricted to the domain $\{0,1\}^n$, denoted $O_n : \{0, 1\}^n \rightarrow \{0, 1\}^n$, $O_n$ is either a $1$-to-$1$ function or a function satisfying Simon's property.
A language which is then contained in $\cBQP^O$ but not in $\cBPP^O$ is $L(O) = \{ 0^n | O_n \text{ is a function with Simon's property} \}$ as shown in \cite{simon}. In fact, as we've mentioned before, the complement of this language\footnote{Note that Simon's problem is a promise problem, so when speaking about the complement of $L(O)$ we are in fact referring to $L^c(O) =\{ 0^n | O_n \text{ is a 1-to-1 function} \}$.} can be used to separate $\cBQP^O$ and $\cNP^O$ \cite{aaronsonph}. Lemma~\ref{lemma:bqpnp}, which we prove below, is essentially a proof of this fact for a slightly different version of the oracle.

Before proving Lemma~\ref{lemma:bqpnp}, let us first address a technical point.
As we remarked in the introduction, an unconstrained GES for $\cBQP$ would imply $\cBQP \in \gesdec$.
Therefore, we would ideally like to construct an oracle to separate $\cBQP^O$ and $\cNPpoly^O$. The intuition in constructing this hypothetical oracle would be following: instead of considering a function $O_n$ for each input length $n$, we consider a function $O_x$, for each input string $x \in \{0, 1\}^n$. In other words, for a fixed input length, $n$, there will be $2^n$ functions which need to be decided. But the $\cNP^O$ machine receives only a polynomial amount of advice, which is the same for all of these $2^n$ functions. Therefore this advice should be insufficient to help the $\cNP^O$ machine in deciding all of these inputs.
Formalising this intuition for any polynomial is problematic, as will become clear later (see the last paragraph of the proof of Lemma~\ref{lemma:bqppnd}). For this reason, we will fix the degree of the polynomial and prove that $\cBQP^O \not\subseteq (\cNPnd)^O$.
To do this, let us first prove the separation between $\cBQP^O$ and $\cNP^O$, for our specific oracle.
\begin{lemma} \label{lemma:bqpnp}
There exists an oracle $O$, based on the complement of Simon's problem, such that $\cBQP^O \not\subseteq \cNP^O$.
\end{lemma}
\begin{proof}
The separation of $\cBQP$ and $\cNP$ with respect to an oracle has been shown a number of times before, \cite{bv, qrfs, watroussep}, including with the complement of Simon's problem. However, we prove this lemma for our particular version of Simon's problem where instead of assigning a function to each input length, we assign different functions to different inputs.

We proceed by defining an oracle $O$ and a language which we refer to as the complement of Simon's problem or $coSimon(O)$, such that $coSimon(O) \in \cBQP^{O}$ and $coSimon(O) \not\in \cNP^{O}$.
We start with the latter as it also clarifies what the oracle should do:
\begin{equation} \label{eqn:cosimon}
coSimon(O) = \{ \Braket{1^n, i} | i \in \{0,1\}^n \text{ and } f(x) \equiv O(1^n, i, x) \text{ is a 1-to-1 function} \}
\end{equation}
Strictly speaking, the problem we are defining is a promise problem, so the set defined above is the set of $yes$ instances to the problem, whereas the set of $no$ instances is not the complement but the set:
\begin{equation}
\{ \Braket{1^n, i} | i \in \{0,1\}^n \text{ and } f(x) \equiv O(1^n, i, x) \text{ is a Simon function} \}
\end{equation}
Here, by ``\emph{Simon function}'' we mean a function having Simon's property.

It is clear from this definition that the oracle $O$ is the one providing the functions for which we want test whether they are $1$-to-$1$ or have Simon's property. Of course, the whole point is to restrict access to the descriptions of those functions and force the algorithm solving the problem to perform queries to the oracle.
It is also clear that for any such $O$, $coSimon(O)$ will be contained in $\cBQP^{O}$ since we can just run Simon's algorithm on the given input and flip acceptance and rejection. As is standard in quantum query complexity, we assume that the behaviour of the quantum oracle is to perform the unitary operation $\Ket{1^n}\Ket{i}\Ket{x}\Ket{y} \rightarrow^O \Ket{1^n}\Ket{i}\Ket{x}\Ket{O(1^n, i, x) \oplus y}$.

The oracle $O$ can be viewed as some function taking as input the tuple $(n, i, x)$ and outputting $f_i(x)$, where $f_i : \{0, 1\}^n \rightarrow \{0, 1\}^n$ is a function which is either bijective or has Simon's property. Essentially $n$, which is given in unary, specifies the domain size of our functions, $i$ is an index for a particular function and $x$ is the value on which we evaluate $f_i$. These last two elements of the tuple are specified in binary and the oracle should be defined for all $n \in \mathbb{N}$ and all $i, x \in \{0, 1\}^n$.
We will denote the set of functions used by the oracle for domain size $n$ as $\mathcal{F}_n$, in other words:
\begin{equation}
\mathcal{F}_n = \{ f_i | i \in \{0, 1\}^n \text{ and } f_i \text{ is defined as } f_i(x) \equiv O(1^n, i, x)  \}
\end{equation}

Next, we construct a so-called \textit{adversarial oracle} $O$. This just means defining the family of sets $ \{ \mathcal{F}_n \}_{n \in \mathbb{N}}$, in such a way that every non-deterministic Turing machine using the oracle $O$ fails to decide correctly $coSimon(O)$. The proof will use a diagonalisation argument.

Since the set of non-deterministic Turing machines is countable we consider the $k$'th machine, $M_k$, and check its behaviour when $n = k + n_0$, for some $n_0 \geq 0$ which we define later on. Suppose we take some index $i \in \{0, 1\}^n$, and tentatively make the $i$'th function in $\mathcal{F}_k$ a $1$-to-$1$ function. By simulating the behaviour of $M_k$ on this input we can check to see whether it accepts or rejects. If it rejects, then we are done, since $M_k$ will incorrectly decide this input. Conversely, if $M_k$ accepts, then by definition there exists a polynomial-sized path, in the non-deterministic computation tree of the machine, which leads to acceptance. We denote this path as $\pi$, and denote the length of $\pi$ as $l = poly(n)$. $M_k$ can make at most $l$ queries to $O$ on this path which we can represent as a list of tuples: $[ (x_1, f_i(x_1)), (x_2, f_i(x_2)) ... (x_l, f_i(x_l)) ]$, where $x_1, ... x_l$ are the queried variables. An example of such a path is shown in Figure~\ref{fig:tree}.

\begin{figure}[htbp!]
\centering
\includegraphics[scale=0.35]{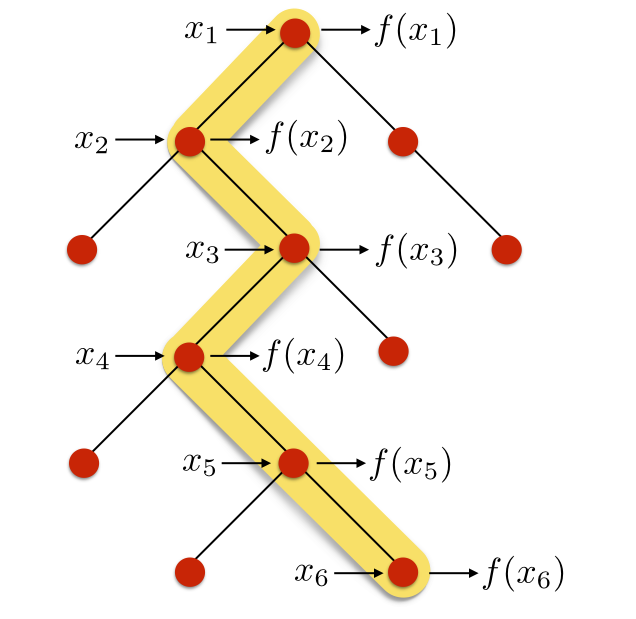}
\caption{Computation tree with queries}
\label{fig:tree}
\end{figure}

We now simply consider a Simon function $f'$ that matched $f_i$ on the queried values, i.e. $f'(x_1) = f_i(x_1), ... f'(x_l) = f_i(x_l)$. How do we know such a function exists? The number of possible bit masks $s$ such that $f(x) = f(x \oplus s)$ is $2^{n}-1$ (since $0^n$ is excluded). By having $f'$ match $f_i$ on the $l$ queried values it must be that $f'$ produces different outputs for each of these values. Therefore for any $i, j \leq l$, $i \neq j$ it must be that $s \neq x_i \oplus x_j$. This means that there are $l(l-1)/2$ values of $s$ which are restricted. But $l = poly(n)$ and since $s$ can take on $2^{n}-1$ possible values, if $n$ is sufficiently large so that $2^{n}-1 > l(l-1)/2$, then we can simply choose an $s$ which is not restricted. We therefore pick $n_0$ to be large enough so that $2^{n} -1 > l(l-1)/2$ and then take $s$ to be some mask from the available $2^{n} - 1 - l(l-1)/2$.
We thus have a Simon function which produces the same responses to the queries on path $\pi$ as the $1$-to-$1$ function $f_i$. If we now just take $f_i$ to be $f'$, then $\pi$ will still be an accepting path and therefore $M_k$ will decide incorrectly on the input $\Braket{1^n, i}$.

Through this construction, all non-deterministic Turing machines will have some input on which they decide $coSimon(O)$ incorrectly, thus $coSimon(O) \not\in \cNP^{O}$ concluding the proof.
\end{proof}

\begin{lemma} \label{cor:oraclema}
There exists an oracle $O$, based on the complement of Simon's problem, such that $\cBQP^{O} \not\subseteq \cMA^O$.
\end{lemma}
\begin{proof}
The arguments from Lemma~\ref{lemma:bqpnp} can be used to show that even if the deciding machine receives a polynomial amount of randomness, it still cannot correctly input (with high probability). This corresponds to showing that the complement of Simon's problem also lies outside of $\cMA$, relative to the oracle.

The idea for this case will be to pick the oracle at random and then reduce the problem to the case without randomness. Suppose the oracle is a $1$-to-$1$ function or a Simon function with equal probability (in either case, the specific function that is picked is chosen uniformly at random).
An $\cMA$ algorithm is essentially a probability distribution over $\cNP$ algorithms. If the complement of Simon's problem, with respect to the random oracle, is in $\cMA$, then there must be an $\cNP$ machine that decides the problem correctly with probability at least $2/3$, over the random choice of the oracle.

We have therefore reduced the task of showing that the problem is not in $\cMA$ to that of showing that it is not in $\cNP$. Consider a non-deterministic Turing machine that accepts a particular input (of the form described in Equation~\ref{eqn:cosimon}) of length $2n$.
If the machine accepts this input that means that there exists an accepting non-deterministic path, making $l$ queries, where $l = poly(n)$. Assuming the function is a Simon function, from the proof of Lemma~\ref{lemma:bqpnp}, it's clear that the probability of finding a collision after $l$ queries , (and thus distinguishing the function from a $1$-to-$1$ function), assuming one has not been found after $l-1$ queries, is:
\begin{equation}
\frac{l - 1}{2^n - 1 - l(l-1)/2}
\end{equation}
But since $l = poly(n)$, this probability is exponentially small in $n$. Since, for any given input, there is an equal chance of it being a $1$-to-$1$ function or a Simon function, it follows that the probability that the algorithm accepts correctly is at most $1/2 + 2^{-\Omega(n)}$, which will be smaller than $2/3$, for sufficiently large $n$.
\end{proof}

\noindent Next, we prove:
\begin{lemma} \label{lemma:bqppnd}
For each $d \in \mathbb{N}$, there exists an oracle $O$, such that $\cBQP^O \not\subset (\cPnd)^O$.
\end{lemma}
\begin{proof}
To begin with, the class $\cPnd$ is the class of problems solved by a deterministic polynomial-time Turing machine $M$, which receives an advice of length $\mathsf{O(n^d)}$, when the input is of size $\mathsf{O(n)}$ (in our case the input size is $2n$ since we defined $n$ as being the length of inputs to the $1$-to-$1$ and Simon functions).

In contrast to the previous case, instead of having the ability to non-deterministically choose one of exponentially many paths, a polynomial-time Turing machine $M$ receives some non-uniform information to help it in deciding $coSimon(O)$.
Each advice determines a new behaviour for $M$ which can even involve a different sequence of queries to the oracle.  What we want to show is that irrespective of what advice $M$ might receive, it still cannot always correctly decide $coSimon(O)$. To do this, we consider functions over a larger domain than just $n$-bit strings. In other words, for each $d$ we choose $D > d$ such that the set $\mathcal{F}_n$ contains $2^n$ functions of the form $f : \{0, 1\}^{n^D} \rightarrow \{0, 1\}^{n^D}$. The oracle, which we now denote as $O_d$, still receives queries of the form $(1^n, i, x)$, where $|i| = n$, but now $|x| = n^{D}$.

First we need to argue that the problem can still be decided in $\cBQP^{O_d}$. This is indeed the case, since expanding the domains of the functions simply changes the running time of the quantum algorithm from $\mathsf{O(n)}$ to $\mathsf{O(n^{D})}$. But since $D$ is just a fixed constant, the algorithm still runs in polynomial time, hence $coSimon(O_d) \in \cBQP^{O_d}$.

The harder part is showing $coSimon(O_d) \not\in \cPnd$.
As before, we will prove this by diagonalisation by considering the set of all (deterministic) Turing machines and showing that no matter which advice the $k$'th machine receives it cannot correctly decide $coSimon(O_d)$. Care must be taken, as each advice induces a different behaviour and one must consider the oracle so that all possible advice strings lead to failure. This is in contrast with the previous case where we were only interested in the behaviour of one accepting path of the non-deterministic computation tree.

Suppose we take the $k$'th deterministic polynomial-time Turing machine, $M_k$, and examine what happens for an input of length $n = k + n_0$, where $n_0$ will be chosen later (as before).
Since the advice is a binary string of length $\mathsf{O(n^d)}$ there are $2^{\mathsf{O(n^d)}}$ possible advice strings.
Whichever one $M_k$ uses it will be the same for all $2^n$ inputs of length $n$.

Let us now consider the first index of length $n$, namely $0^n$ and assign a $1$-to-$1$ function $f : \{0,1\}^{n^D} \rightarrow \{0,1\}^{n^D}$ to this index. We can inspect the behaviour of $M_k$ for $f$ and for each possible advice string. If for more than half of the advice strings $M_k$ rejects, then we keep $f$ at index $0^n$. This means that half of all advice strings have been eliminated (there is at least one input on which those strings lead to $M_k$ deciding incorrectly).
If, however more than half of all advice strings make $M_k$ accept $f$, we will attempt to turn $f$ into a Simon function while keeping acceptance for those advice strings. This will again lead to the elimination of (at least) half of all advice strings.

For each advice $a_j$, where $1 \leq j \leq 2^{\mathsf{O(n^d)}}$,  $M_k$ will make a sequence of polynomially many queries to $f$.
Denote that sequence of queries together with the responses as:
\[
\sigma_j = [ (x_{1j}, f(x_{1j})); (x_{2j}, f(x_{2j})); ... (x_{lj}, f(x_{lj})) ]
\]
where $l = poly(n)$.
We now consider a Simon function $f' : \{0,1\}^{n^D} \rightarrow \{0,1\}^{n^D}$ such that for all $j$ in which $M_k$ with advice $a_j$ and queries $\sigma_j$ accepts and for all $t \leq l$, we have that $f'(x_{tj}) = f(x_{tj})$. In other words $f'$ will give identical responses to the queries which make $M_k$ accept.
Since $t$ ranges from $1$ to $l = poly(n)$ and $j$ ranges from $1$ to $2^{\mathsf{O(n^d)}}$, the maximum number of variables which are queried is of order $2^{\mathsf{O(n^d)}}$.
But unlike in the previous lemma, this number is exponential in the size of the input, so how can we be sure that such a Simon function even exists?
The trick is that we can choose the domain size through $D$ and make it large enough to accommodate for a Simon function with this property.

As before, because $f$ is bijective, no two queried variables will produce the same answer. Therefore, there cannot be a bit mask $s$ ($|s| = n^D$) relating any pair of the $2^{\mathsf{O(n^d)}}$ queries. These will be the restricted values of $s$. The total number of such values is also of order $2^{\mathsf{O(n^d)}}$, however the total number of possible values is $2^{\mathsf{O(n^D)}}$. Thus, if we simply choose $D$ such that $2^{\mathsf{O(n^D)}} > 2^{\mathsf{O(n^d)}}$ then we can find a Simon function $f'$ which matches the responses of $f$ on the $2^{\mathsf{O(n^d)}}$ queries.

Hence, for this case if we use $f'$ as the function for index $0^n$ we will eliminate half of the possible advice strings. Thus, no matter how $M_k$ behaves we are able to eliminate half of all possible advice strings with our first input of length $2n$.
Clearly this process can be repeated for the next index and so on until the last index. We are effectively halving the number of potentially useful advice strings with each index. Since we are doing this $2^n$ times, to eliminate all possible advice strings we just need to ensure that $2^{\mathsf{O(n^d)}} / 2^{2^n} < 1$ or $2^{\mathsf{O(n^d)}} < 2^{2^n}$. To achieve this, simply choose $n_0$ (recall that $n = k + n_0$) large enough so that the inequality holds.

We therefore have that for all $k$, and for all possible advice strings, there will always be an input to $coSimon(O_d)$ which is decided incorrectly, hence $coSimon(O_d) \not\in \cPnd$.

Note that the same proof would not work for $\cPpoly$. A crucial element in our proof was the fact that we can make $D$ (which determines the size of the domain of each function) to be much larger than $d$ (which determines the length of the advice).
But this is only possible because $d$ is fixed from the very beginning. If the advice length could be any arbitrary polynomial then no matter what constant value of $D$ we decided upon for our oracle, there would always be some $d > D$ and hence some polynomial length of the advice string for which the proof does not work.
A possible ``fix'' would be to make $D$ part of the input in some form, so that it too can increase. So if, say, $D$ was included in the input as a $g(n)$ unary string, where $g$ is some monotonically increasing function, then for sufficiently large $n$, $g(n) > d$. But we immediately notice the problem with this approach. While it is true that in this case the problem cannot be decided in $\cPpoly^{O}$ it would also no longer be decidable in $\cBQP^O$ either. This is because the query complexity of the quantum algorithm becomes $\mathsf{O(n^{g(n)})}$ which is no longer polynomial unless $g$ is the constant function.
Hence, proving separation from $\cPpoly$ seems to require some non-trivial modification of this proof or a totally different technique.
\end{proof}

\noindent Finally, we can prove Theorem~\ref{thm:noGESBQP} by combining the previous results.

\begin{proof}[Proof of Theorem~\ref{thm:noGESBQP}]
Let us first show that, relative to an oracle, the complement of Simon's problem is not contained in $\mathsf{NP/O(n^d)}$.
The oracle $O_d$ will be defined in the exact same way as for the $\cPnd$ case.
The same reasoning as before applies here. Take the $k$'th non-deterministic Turing machine and examine its behaviour for some input $\Braket{1^n, i}$, where $n = k + n_0$ and $n_0$ is chosen as before.
For each index, we tentatively pick a $1$-to-$1$ function and examine what the machine does for each advice of length $\mathsf{O(n^d)}$. If more than half of the advice strings lead to rejection then we keep the bijective function and proceed to the next index. Otherwise we replace it with a Simon function. In this case, for each advice in which the machine accepts, there will be some polynomial-sized path leading to acceptance.
We will pick one accepting path for each advice on which the machine accepts and ensure that the Simon function produces the same responses to the queries on those paths. This reduces the problem to the previous case. We know that for sufficiently large $D$ such a function exists and therefore each index will render half of the possible advice strings useless.
By also choosing $n_0$ large enough we can make sure that all advice strings are eliminated and thus that the problem is incorrectly decided by all non-deterministic Turing machines irrespective of the advice (of length $\mathsf{O(n^d)}$).

For the $\cNPnd$ case, one can use the same proof as in Lemma~\ref{cor:oraclema} to reduce to the $\mathsf{NP/O(n^d)}$ case.
It follows that $coSimon(O_d) \not\in \cNPnd$.
\end{proof}

\section{GES for exact \textsc{BosonSampling} and circuits for the permanent} \label{sect:proofs2}
\noindent To prove Theorem~\ref{thm:noGESsampBQP}, we first need to show a number of results concerning permanents of matrices. The purpose of these results is to eventually show that having an oracle for estimating the squared permanent of a matrix taking values in $\{-1, 0, 1 \}$, yields a polynomial-time algorithm, with random access to $n^{\sf{O}(n)}$ bits of advice, for exactly computing the permanent.
This result together with the assumption that a GES allows the client to sample exactly from the \textsc{BosonSampling} distribution and a result of Bj{\"o}rklund, from \cite{andreas}, will allow us to prove Theorem~\ref{thm:noGESsampBQP}.

Let us first introduce some helpful notation: for a matrix, $A$, we will denote $A^{i,j}$ as the matrix obtained by deleting row $i$ and column $j$ from $A$.
\begin{lemma} \label{lem:helper1}
Let $X = (x_{i,j}) \in \{-1, 0, 1 \}^{n \times n}$. There exists a matrix $Z = (z_{i,j}) \in \{-1, 0, 1 \}^{(n + 2) \times (n + 2)}$ such that:
\begin{itemize}
\item $z_{n+2,n+2} = 0$
\item $Per(Z) = -Per(X)$
\item $Per(Z^{n+2,n+2}) = Per(X^{1,1})$
\end{itemize}
\end{lemma}
\begin{proof}
Let $Z$ be the following matrix:
\[
Z = \begin{bmatrix} 
	x_{n,n} & x_{n, n-1} & \dots & x_{n,1} & 0 & 0 \\
	x_{n-1,n} & x_{n-1,n-1} & \dots & x_{n-1,1} & 0 & 0 \\
    \vdots & \vdots & & \vdots & \vdots & \vdots \\
	x_{1,n} & x_{1,n-1} & \dots & x_{1,1} & 1 & -1 \\
	0 & 0 & \dots & 1 & 0 & 1 \\
	0 & 0 & \dots & -1 & -1 & 0 \\
    \end{bmatrix}
\]
We can see that $z_{n+2, n+2} = 0$. It is also not difficult to see that $Per(Z^{n+2,n+2}) = Per(X^{1,1})$, through a Laplace expansion.
We now perform a Laplace expansion along the last row of $Z$, to compute its permanent:
\begin{equation}
Per(Z) = - (Per(Z^{n+2, n+1}) + Per(Z^{n+2, n}))
\end{equation}
But $Per(Z^{n+2, n}) = Per(X^{1,1})$ and $Per(Z^{n+2, n+1}) = Per(X) - Per(X^{1,1})$ hence $Per(Z) = -Per(X)$.
\end{proof}

\begin{lemma} \label{lem:helper2}
Let $X = (x_{ij}) \in \{-1, 0, 1 \}^{n \times n}$, $Z = (z_{ij}) \in \{-1, 0, 1 \}^{m \times m}$, $m \geq 2$, such that $z_{m,m} = 0$ and $W = (w_{ij}) \in \{-1, 0, 1 \}^{(m + n -1) \times (m + n - 1)}$ defined as follows:
\[
W = \begin{bmatrix} 
	z_{1,1} & z_{1,2} & \dots & z_{1,m} & 0 & \dots & 0 \\
	z_{2,1} & z_{2,2} & \dots & z_{2,m} & 0 & \dots & 0 \\
    \vdots & \vdots & & \vdots & \vdots & & \vdots \\
	z_{m-1,1} & z_{m-1,2} & \dots & z_{m-1,m} & 0 & \dots & 0 \\
	z_{m,1} & z_{m,2} & \dots & x_{1,1} & x_{1,2} & \dots & x_{1,n} \\
	0 & 0 & \dots & x_{2,1} & x_{2,2} & \dots & x_{2,n} \\	
    \vdots & \vdots & & \vdots & \vdots & & \vdots \\
	0 & 0 & \dots & x_{n,1} & x_{n,2} & \dots & x_{n,n} \\	
    \end{bmatrix}
\]
Then, it is the case that:
\begin{equation}
Per(W) = Per(Z) Per(X^{1,1}) + Per(Z^{m,m}) Per(X)
\end{equation}
\end{lemma}
\begin{proof}
We will prove this by induction over $m$. For the $m=2$ case we have:
\[
W = \begin{bmatrix} 
	z_{1,1} & z_{1,2} & 0 & \dots & 0 \\
	z_{2,1} & x_{1,1} & x_{1,2} & \dots & x_{1,n} \\
	0 & x_{2,1} & x_{2,2} & \dots & x_{2,n} \\	
    \vdots & \vdots & \vdots & & \vdots \\
	0 & x_{n,1} & x_{n,2} & \dots & x_{n,n} \\	
    \end{bmatrix}
\]
By doing a Laplace expansion along the first row of $W$, we get:
\begin{equation}
Per(W) = z_{1,1} Per(X) + z_{1,2} z_{2,1} Per(X^{1,1})
\end{equation}
Now note that:
\[
Z = \begin{bmatrix}
	z_{1,1} & z_{1,2} \\
	z_{2,1} & 0 \\
\end{bmatrix}
\]
So $Per(Z) = z_{1,2} z_{2,1}$ and $Per(Z^{2,2}) = z_{1,1}$, therefore:
\begin{equation}
Per(W) = Per(Z) Per(X^{1,1}) + Per(Z^{2,2}) Per(X)
\end{equation}
We now assume the relation is true for $m - 1$ and prove it for $m$. To do this, we will first Laplace expand the permanent of $W$ along the first row:
\begin{equation}
Per(W) = \sum\limits_{i=1}^{m-1} z_{1,i} Per(W^{1,i}) + z_{1,m} Per(W^{1,m})
\end{equation}
The reason for separating the terms this way, is because $W^{1,i}$, with $i < m$, is of the same form as $W$ and we can therefore apply the induction hypothesis. Doing so yields:
\begin{equation}
Per(W) = \sum\limits_{i=1}^{m-1} z_{1,i} (Per(Z^{1,i}) Per(X^{1,1}) +  Per(Z^{(1,m),(i,m)}) Per(X)) + z_{1,m} Per(W^{1,m})
\end{equation}
Where $Z^{(1,m),(i,m)}$ is obtained from $Z$ by deleting rows $1$ and $m$ and columns $i$ and $m$.
Taking common factors we get:
\begin{equation}
Per(W) = Per(X^{1,1}) \sum\limits_{i=1}^{m-1} z_{1,i}Per(Z^{1,i})
+  Per(X)\sum\limits_{i=1}^{m-1} z_{1,i}Per(Z^{(1,m),(i,m)}) + z_{1,m} Per(W^{1,m})
\end{equation}
But notice that:
\begin{equation}
\sum\limits_{i=1}^{m-1} z_{1,i}Per(Z^{(1,m),(i,m)}) = Per(Z^{m,m})
\end{equation}
since it is a Laplace expansion along the first row of $Z^{m,m}$. This leads to:
\begin{equation} \label{eqn:intermediate}
Per(W) = Per(X^{1,1}) \sum\limits_{i=1}^{m-1} z_{1,i}Per(Z^{1,i}) + Per(X)Per(Z^{m,m}) + z_{1,m} Per(W^{1,m})
\end{equation}

\noindent The matrix $W^{1,m}$ is of the same form as $W$:
\[
W^{1,m} = \begin{bmatrix} 
	z_{2,1} & z_{2,2} & \dots & z_{2,m - 1} & 0 & \dots & 0 \\
    \vdots & \vdots & & \vdots & \vdots & & \vdots \\
	z_{m-1,1} & z_{m-1,2} & \dots & z_{m-1,m - 1} & 0 & \dots & 0 \\
	z_{m,1} & z_{m,2} & \dots & z_{m,m-1} & x_{1,2} & \dots & x_{1,n} \\
	0 & 0 & \dots & 0 & x_{2,2} & \dots & x_{2,n} \\	
    \vdots & \vdots & & \vdots & \vdots & & \vdots \\
	0 & 0 & \dots & 0 & x_{n,2} & \dots & x_{n,n} \\	
    \end{bmatrix}
\]
We can see this by taking:
\[
Z_{W^{1,m}} = 
\begin{bmatrix} 
	z_{2,1} & z_{2,2} & \dots & z_{2,m - 1}  \\
    \vdots & \vdots & & \vdots \\
	z_{m-1,1} & z_{m-1,2} & \dots & z_{m-1,m - 1} \\
	z_{m,1} & z_{m,2} & \dots & 0 \\
\end{bmatrix}
\quad \quad
X_{W^{1,m}} = 
\begin{bmatrix} 
	z_{m,m-1} & x_{1,2} & \dots & x_{1,n} \\
	0 & x_{2,2} & \dots & x_{2,n} \\	
    \vdots & \vdots & & \vdots\\
	0 & x_{n,2} & \dots & x_{n,n} \\	
\end{bmatrix}
\]
Together with the induction hypothesis this gives us:
\begin{equation}
Per(W^{1,m}) = Per(Z_{W^{1,m}}) Per(X_{W^{1,m}}^{1,1}) + Per(Z_{W^{1,m}}^{m-1,m-1}) Per(X_{W^{1,m}})
\end{equation}
Now note that $Per(X_{W^{1,m}}) = z_{m,m-1} Per(X_{W^{1,m}}^{1,1})$ and $Per(X_{W^{1,m}}^{1,1}) = Per(X^{1,1})$, hence:
\begin{equation}
Per(W^{1,m}) = Per(X^{1,1})(Per(Z_{W^{1,m}})  + z_{m,m-1} Per(Z_{W^{1,m}}^{m-1,m-1})
\end{equation}
But the term in parenthesis is $Per(Z^{1,m})$ so:
\begin{equation}
Per(W^{1,m}) = Per(X^{1,1})Per(Z^{1,m})
\end{equation}
By substituting this into Equation~\ref{eqn:intermediate}, we get:
\begin{equation}
Per(W) = Per(X^{1,1}) \sum\limits_{i=1}^{m-1} z_{1,i}Per(Z^{1,i})
+ Per(X)Per(Z^{m,m}) + z_{1,m} Per(X^{1,1})Per(Z^{1,m})
\end{equation}
After grouping terms:
\begin{equation}
Per(W) = Per(X^{1,1}) \sum\limits_{i=1}^{m} z_{1,i}Per(Z^{1,i}) + Per(X)Per(Z^{m,m})
\end{equation}
But:
\begin{equation}
\sum\limits_{i=1}^{m} z_{1,i}Per(Z^{1,i}) = Per(Z)
\end{equation}
Thus:
\begin{equation}
Per(W) = Per(Z) Per(X^{1,1}) + Per(Z^{m,m})Per(X)
\end{equation}
This concludes the proof.
\end{proof}

\noindent Using the above lemmas, we can now show the following:
\begin{theorem} \label{thm:gadget}
Let $\mathcal{O}$ be an oracle that, given a matrix $X \in \{-1, 0, 1 \}^{n \times n}$, outputs a number $\mathcal{O}(X)$ such that:
\begin{equation}
\frac{Per(X)^2}{g} \leq \mathcal{O}(X) \leq g Per(X)^2
\end{equation}
where $g \in [1, poly(n)]$. Then, for any $X \in \{-1, 0, 1 \}^{n \times n}$ there exists a polynomial time algorithm for computing $Per(X)$, which has random access to $n^{\mathsf{O}(n)}$ bits of advice and making $poly(n)$ queries to $\mathcal{O}$.
\end{theorem}
\begin{proof}
The theorem shows that having an oracle for computing a multiplicative approximation for the squared permanent of a matrix, implies the existence of a polynomial time algorithm, with $n^{\sf{O}(n)}$ bits of advice, that can compute the permanent exactly.

The proof of this theorem is inspired from a similar result of Aaronson and Arkhipov (see Theorem 4.3 from \cite{bosonsampling}). In that case, the oracle was outputting a multiplicative approximation of the squared permanent of an arbitrary \emph{real} matrix. In our case, however, the matrices are restricted to entries from $\{-1, 0, 1\}$, which means that we cannot directly use that result.

We prove the theorem by induction. For the case of $n=1$ the algorithm simply outputs $X$. Suppose now that we have an algorithm  for computing the permanents of $(n-1) \times (n-1)$ matrices with entries from $\{-1, 0, 1\}$. We will use this algorithm to compute the permanent of $X$. 
Firstly, if $\mathcal{O}(X) = 0$, then $Per(X) = 0$ and we are done. Additionally, we are going to use the oracle to check if any of the $(n - 1) \times (n - 1)$ minors of $X$ are non-zero. If all of them are zero, then $Per(X) = 0$ again and we are done. So let's assume that $Per(X) \neq 0$ and $Per(X^{1,1}) \neq 0$\footnote{The permanent is invariant under permutations of rows and columns. Thus, if $X$ has a non-zero minor, we can simply permute the columns of $X$, so that $X^{1,1}$ is that minor.}. 

We know from Lemma~\ref{lem:helper2}, that if we take a matrix $Z = (z_{ij}) \in \{-1, 0, 1 \}^{m \times m}$, $m \geq 2$, such that $z_{m,m} = 0$ and then construct:
\[
W = \begin{bmatrix} 
	z_{1,1} & z_{1,2} & \dots & z_{1,m} & 0 & \dots & 0 \\
	z_{2,1} & z_{2,2} & \dots & z_{2,m} & 0 & \dots & 0 \\
    \vdots & \vdots & & \vdots & \vdots & & \vdots \\
	z_{m-1,1} & z_{m-1,2} & \dots & z_{m-1,m} & 0 & \dots & 0 \\
	z_{m,1} & z_{m,2} & \dots & x_{1,1} & x_{1,2} & \dots & x_{1,n} \\
	0 & 0 & \dots & x_{2,1} & x_{2,2} & \dots & x_{2,n} \\	
    \vdots & \vdots & & \vdots & \vdots & & \vdots \\
	0 & 0 & \dots & x_{n,1} & x_{n,2} & \dots & x_{n,n} \\	
    \end{bmatrix}
\]
we have that:
\begin{equation}
Per(W) = Per(Z^{m,m}) Per(X) + Per(Z) Per(X^{1,1})
\end{equation}
If $Per(W) = 0$ and $Per(Z^{m,m}) \neq 0$, then:
\begin{equation} \label{eqn:recursion}
Per(X) = - Per(X^{1,1}) \frac{Per(Z)}{Per(Z^{m,m})}
\end{equation}
From Lemma~\ref{lem:helper1}, we know that for any $n \times n$ matrix $X$, there exists an $(n+2) \times (n+2)$ matrix $Z$, such that $z_{n+2,n+2} = 0$, $Per(Z) = - Per(X)$ and $Per(Z^{n+2, n+2}) = Per(X^{1,1})$. If one used such a $Z$ in the construction of $W$, then it is immediate that $Per(W) = 0$ and that $Per(Z^{n+2,n+2}) \neq 0$.
The algorithm will search for such a $Z$, construct the corresponding $W$ and use the oracle to test if $Per(W) = 0$. If the permanent of $W$ is zero, then one can compute the permanent of $X$ using Equation~\ref{eqn:recursion}.

But how do we search for $Z$ and, furthermore, how do we compute $Per(Z)/Per(Z^{n+2,n+2})$? This is where we make use of advice. Note that since $Z \in  \{-1, 0, 1 \}^{(n+2) \times (n+2)}$, we have that:
\begin{equation}
-(n+2)! \leq Per(Z) \leq (n+2)!
\end{equation}
hence, there are at most $n^{\mathsf{O}(n)}$ possible values for the permanents of $Z$ matrices. Similarly, there are at most $n^{\mathsf{O}(n)}$ possible values for the permanents of the $Z^{n+2,n+2}$ minors of $Z$ matrices.

The advice, to our algorithm, will consist of tuples $(Z_i, Per(Z_i^{n+2,n+2}), f_i = Per(Z_i)/Per(Z_i^{n+2,n+2}))$, comprising of a matrix $Z_i$ together with the permanent of its top left $(n+1) \times (n+1)$ minor and the ratio between that matrix's permanent and the permanent of its top left minor, with $i \leq n^{cn}$, for some constant $c > 0$. Here $Z_i \in  \{-1, 0, 1 \}^{(n+2) \times (n+2)}$, with the bottom right entry being $0$ and $Per(Z_i^{n+2,n+2}) \neq 0$. The matrices in the tuples are such that all possible values for the fraction $f_i = Per(Z_i)/Per(Z_i^{n+2,n+2})$ are covered. From the above discussion, it's clear that there will be at most $n^{\mathsf{O}(n)}$ such tuples. Furthermore, the tuples are sorted in ascending order with respect to those fractions.

For a given matrix $X$, our algorithm should search through this advice in order to find a matrix $Z_i$ such that $\mathcal{O}(W_i) = 0$, where $W_i$ is constructed from $X$ and $Z_i$ as in Lemma~\ref{lem:helper2}. When such a matrix is found, we have that:
\begin{equation}
Per(X) = -Per(X^{1,1}) f_i
\end{equation}
But $f_i$ is given in the advice tuple and $Per(X^{1,1})$ is computed recursively by our algorithm, hence we have computed the permanent of $X$.

To find the matrix $Z_i$ we will perform a binary search over the advice. Suppose that $i$ ranges from $1$ to $l = n^{\mathsf{O}(n)}$. Additionally, let $\alpha_i = Per(Z_i^{n+2,n+2})$, so that:
\begin{equation}
Per(W_i) = \alpha_i (Per(X) + f_i Per(X^{1,1}))
\end{equation}
This means that computing $\sqrt{\mathcal{O}(W_i)}/|\alpha_i|$ gives us a multiplicative approximation for $|Per(X) + f_i Per(X^{1,1})|$. Because the $f_i$ values are sorted in ascending order, this means that the function:
\begin{equation}
h(i) = Per(X) + f_i Per(X^{1,1})
\end{equation}
is monotonically increasing as a function of $i$ and furthermore that there is a unique value $i$ such that $h(i) = 0$. In our case, however, we have a multiplicative approximation for $|h(i)|$, which we denote as $t(i) = \sqrt{\mathcal{O}(W_i)}/|\alpha_i|$. This function will be monotonically decreasing between $1$ and $i$ and increasing between $i$ and $l$. 
We look for $i$ using binary search as follows: compute $t(v)$ and $t(w)$ for the middle two points, $v$ and $w$ of the interval $[1,l]$. If either of them is $0$, then we are done. Otherwise, if $t(v) < t(w)$, then search the interval $[2,v]$, otherwise the interval $[w,l]$. Repeat this recursively until the minimum is found.

Given that the advice is of length $n^{\mathsf{O}(n)}$, the algorithm will query it (and consequently $\mathcal{O}$ as well) at most $\mathsf{O}(n \log n)$ times. Additionally, the construction of each $W_i$ takes time $\mathsf{O}(n^2)$ and since this is done at most $\mathsf{O}(n \log n)$ times, the complexity of this step is $\mathsf{O}(n^3 \log n)$. Finally, the algorithm performs recursive calls to itself (in order to compute $Per(X^{1,1})$) and if we add up the running times of each step we find that the total runtime will be $\mathsf{O}(n^4 \log n)$.
\end{proof}

\begin{theorem} \label{thm:sampcirc}
If a $\cBPPrpoly$ algorithm can sample exactly from the \textsc{BosonSampling} distribution then for any matrix $X \in \{-1, 0, 1 \}^{n \times n}$, there exist circuits of size $2^{n - \mathsf{\Omega} \left(\frac{n}{\log n} \right)}$, making polynomially-sized queries to an $\cNP$ oracle, for computing $Per(X)$.
\end{theorem}
\begin{proof}
The starting point for our proof is a result by Bj{\"o}rklund \cite{andreas}. He showed that, for $k \leq n$, the permanent of an $n \times n$ matrix, $X$, can be expressed as a linear combination of $poly(n) 2^{n - k}$ permanents of $k \times k$ matrices. It should be noted that these matrices are not necessarily minors of the original matrix. Nevertheless, each $k \times k$ matrix can be computed efficiently given the input matrix, $X$.

Our task will be to compute all of these $poly(n) 2^{n - k}$ permanents and then perform the linear combination so as to arrive at the permanent of $X$. 
We will use the result of Theorem~\ref{thm:gadget} together with the fact that there exists a $\cBPPrpoly$ algorithm for exact \textsc{BosonSampling}, to show that the permanent of a $k \times k$ matrix can be computed in polynomial time, using random access to $k^{\mathsf{O}(k)}$ bits of advice and polynomially-sized queries to an $\cNP$ oracle.
Crucially, the $k^{\mathsf{O}(k)}$-sized advice will be the same for all $k \times k$ matrices. This means that all permanents can be computed in $poly(n) 2^{n - k}$ time with access to $k^{\mathsf{O}(k)}$ bits of advice.
The explicit value of $k$, as a function of $n$, will be chosen later.

Consider a $k \times k$ matrix, $M$, and a value $\epsilon > 0$, to be chosen later. We embed $\epsilon M$, a scaled version of $M$, as a submatrix of a \textsc{BosonSampling} matrix $A_M$. In other words, $A_M \in \mathbb{C}^{m \times k}$, with $m = poly(k)$ (see \cite{bosonsampling} for more details).
We then have that the probability of detecting one photon in each output mode, a state which we denote as $\ket{1}$, is:
\begin{equation}
p = Per(\epsilon M)^2 = {\epsilon}^{2k} \cdot Per(M)^2
\end{equation}
Since $Per(M) \leq k!$, to ensure that $p \leq 1$, it suffices to set $\epsilon = 1/k$.

If a $\cBPPrpoly$ machine can simulate the output distribution of a \textsc{BosonSampling} instance, $A_M$, that means that:
\begin{equation} \label{eqn:p}
\sum_{y} q_y Pr( \mathcal{A}(A_M, y) \text{ outputs } \ket{1}) = p
\end{equation}
where $\mathcal{A}$ is a $\cBPP$ algorithm and $y$ is the $\mathsf{rpoly}$ advice string, of size polynomial in $k$, drawn from the distribution $\mathcal{D}_{k} = \{ q_y \}_y$. Note that $\mathcal{D}_k$ only depends on $k$.
If we can estimate $p$ up to multiplicative error in polynomial time (potentially with the help of an $\cNP$ oracle and $k^{\mathsf{O}(k)}$ bits of advice) then we will effectively be simulating the oracle $\mathcal{O}$ from Theorem~\ref{thm:gadget}.

To do this, first note that if $Per(M) \neq 0$, the smallest possible value of $p$ is $1/k^{\mathsf{O}(k)}$. We will therefore consider our advice string to consist of $k^{\mathsf{O}(k)}$ samples from $\mathcal{D}_k$, along with their associated probabilities\footnote{The fact that the advice has size $k^{\mathsf{O}(k)}$ will ensure that, if $p$ is non-zero, then a $y$ such that $Pr( \mathcal{A}(A_M, y) \text{ outputs } \ket{1}) > 0$ is overwhelmingly likely to be sampled when we generate the advice.}. Denote the set of these samples as $S$. This allows us to define:
\begin{equation}
p_{est} = \sum_{y \in S} q_y Pr( \mathcal{A}(A_M, y) \text{ outputs } \ket{1}) 
\end{equation}
as a multiplicative estimate for $p$. But $\mathcal{A}$ is a $\cBPP$ algorithm, which means that we can view it as a polynomial-time function, $f_{\mathcal{A}}(A_M, y, r)$ which receives as input (apart from $A_M$ and $y$) a random string $r \in \{0, 1\}^{l(k)}$, for some polynomial $l$. The function will output either $1$, corresponding to the cases where $\mathcal{A}$ outputs $\ket{1}$, or $0$, corresponding to the cases where $\mathcal{A}$ produces some other output. We therefore have that:
\begin{equation}
Pr( \mathcal{A}(A_M, y) \text{ outputs } \ket{1}) = \frac{1}{2^{l(k)}} \sum_{r \in \{0, 1\}^{l(k)}} f_{\mathcal{A}}(A_M, y, r)
\end{equation}
and this leads to our estimate of $p$ being:
\begin{equation}
p_{est} = \frac{1}{2^{l(k)}} \sum_{y \in S} \sum_{r \in \{0, 1\}^{l(k)}} q_y f_{\mathcal{A}}(A_M, y, r)
\end{equation}

Computing $p_{est}$ exactly would require summing exponentially many terms. However notice that $p_{est}$ is a sum of exponentially many \emph{positive} numbers, each of which can be evaluated in polynomial time (given access to the $k^{\mathsf{O}(k)}$ advice). 
For this reason, we can use Stockmeyer's approximate counting method to compute a multiplicative estimate of $p_{est}$ \cite{stockmeyer}. This will, of course, also yield a multiplicative estimate for $p$ itself.

We have thus given an algorithm for computing a multiplicative estimate of a $k \times k$ matrix $M$ that works in time polynomial in $k$, performs queries to an $\cNP$ oracle and has random access to $k^{\mathsf{O}(k)}$ bits of advice.
This algorithm can now be viewed as an implementation of the oracle $\mathcal{O}$ from Theorem~\ref{thm:gadget}.
Using that theorem, leads to a polynomial-time algorithm, with access to an $\cNP$ oracle and $k^{\mathsf{O}(k)}$-size advice, for computing $Per(M)$ \emph{exactly}.
However, since the advice is the same for all $k \times k$ matrices, by repeating this procedure for all $poly(n) 2^{n - k}$ $k \times k$ matrices and combining the results we obtain an algorithm for computing $Per(X)$ that runs in time $poly(n) 2^{n - k}$, uses $k^{\mathsf{O}(k)}$ bits of advice and makes polynomially-sized queries to an $\cNP$ oracle.

The last step is to convert this algorithm into a circuit. Since the algorithm has a running time of $poly(n) 2^{n - k}$, by choosing $k = c \; n / \log n$, for some suitable constant $c > 0$, we will have circuits of size at most $2^{n - \Omega \left( \frac{n}{log(n)} \right)}$. These circuits, must also operate on the $k^{\mathsf{O}(k)}$ bits of advice. Note that $k^{\mathsf{O}(k)} \ll 2^{n - \Omega \left( \frac{n}{log(n)} \right)}$. To reproduce the random access to these bits, we will assume that the gates have unbounded fan-in. The advice bits are therefore hardcoded into the circuit and ``fed'' into each part of the algorithm that makes use of them. Since only polynomially-many bits of the advice are used at any given point, this can only increase the size of our circuit by a polynomial factor.
This concludes the proof.
\end{proof}

With the above result, the proof of Theorem~\ref{thm:noGESsampBQP} is immediate:
\begin{proof}[Proof of Theorem~\ref{thm:noGESsampBQP}]
Recall from the end of Section~\ref{sect:crypto} that the existence of a GES for sampling from a distribution, $\mathcal{D}_x$, implies the existence of a probabilistic polynomial-time algorithm, with an $\cNP$ oracle, and which receives randomised polynomial-sized advice, for sampling from $\mathcal{D}_x$. In other words, having a GES for \textsc{BosonSampling} means having a $\sf{BPP^{NP}/rpoly}$ algorithm that can sample from the exact \textsc{BosonSampling} distribution.
We now notice that the result of Theorem~\ref{thm:sampcirc} relativises. In particular, this means that if
a $\sf{BPP^{NP}/rpoly}$ algorithm can sample from the exact \textsc{BosonSampling} distribution, then for any matrix $X \in \{-1, 0, 1 \}^{n \times n}$, there exist circuits of size $2^{n - \mathsf{\Omega} \left(\frac{n}{\log n} \right)}$, making polynomially-sized queries to an $\cNP^\cNP$ oracle, for computing $Per(X)$.
\end{proof}

\section{Quantum GES} \label{sect:qges}
\subsection{An upper bound for QGES functions}
This section is, in a certain sense, dedicated to `quantising' the Abadi et al.\ result. First of all, we need to define the quantum generalised encryption scheme. Since we're defining this in analogy to UBQC and other blind quantum computing protocols, our QGES considers a single quantum message from the client to the server, rather than allowing the entire interaction between client and server to be quantum. More formally, we have the following:
\begin{definition}
[Quantum Generalised Encryption Scheme (QGES)] \label{defn:QGES}
A quantum generalised encryption scheme (QGES) is a two party protocol between a quantum polynomial-time client $C$, and an unbounded server $S$, characterised by:
\begin{itemize}
\item A function $f: \{0, 1\}^* \rightarrow \{0, 1\}$.
\item A cleartext input $x \in \{0, 1\}^*$, for which the client wants to compute $f(x)$.
\item An expected polynomial-time key generation algorithm $K$ which works as follows: for any $x \in \{0, 1\}^*$, with probability greater than $1/2 + 1/poly(|x|)$ we have $(k, success) \leftarrow K(x)$, where $k \in \{0,1\}^{poly(|x|)}$. If the algorithm does not return $success$ then we have $(k', fail) \leftarrow K(x)$, where $k' \in \{0,1\}^{poly(|x|)}$.
\item A quantum polynomial-time algorithm $QE$, that takes as input classical bit strings and produces as output a quantum state, which works as follows: for any $x \in \{0, 1\}^*$, $k \in \{0, 1\}^{poly(|x|)}$ we have that $\Ket{y} \leftarrow QE(x, k)$, where $\Ket{y} \in \mathcal{H}$ and $dim(\mathcal{H}) = 2^{poly(|x|)}$.
\item A polynomial-time deterministic algorithm $E$ which works as follows: for any $x \in \{0, 1\}^*$, $k \in \{0, 1\}^{poly(|x|)}$ and $s \in \{0,1\}^{poly(|x|)}$ we have that $w \leftarrow E(x, k, s)$, where $w \in \{0,1\}^{poly(|x|)}$.
\item  A polynomial-time deterministic decryption algorithm $D$, which works as follows: for any $x \in \{0, 1\}^*$, $k \in \{0, 1\}^{poly(|x|)}$ and $s \in \{0,1\}^{poly(|x|)}$ we have that $z \leftarrow D(s, k, x)$, where $z \in \{0,1\}^{poly(|x|)}$.
\end{itemize}
And satisfying the following properties:
\begin{enumerate}
\item There are $m$ rounds of communication, such that $m = poly(|x|)$. Denote the client's message in round $i$ as $c_i$ and the server's message as $s_i$.
\item On cleartext input $x$, $C$ runs the key generation algorithm until success to compute a key $(k, success) = K(x)$. This happens before the communication between $C$ and $S$ is initiated, and the key $k$ is used throughout the protocol.
$C$ then runs $QE(x,k)$ to obtain a quantum encryption of the input, $\Ket{y}$ and sends it to $S$\footnote{It should be noted that it makes no difference for our definition if the client sends the whole state $\ket{y}$ to the server or part of it. The state received by the server will be mixed and the only important property we require is that the client has a purification of this state.}.
\item In round $i$ of the protocol, $C$ computes $c_i = E(x, k, \overline{s}_{i-1})$, where $\overline{s}_{i-1}$ denotes the server's responses up to and including round $i-1$, i.e.\ $\langle s_0, s_1 ... s_{i-1} \rangle$. We assume that $s_0$ is the empty string. $C$ then sends $c_i$ to $S$.
\item In round $i$ of the protocol, $S$ responds with $s_i$, such that $s_i \in \{0,1 \}^{poly(|x|)}$. 
\item At the end of the protocol, $C$ computes $z \leftarrow D(\overline{s}_m, k, x)$ and with probability $1/2 + 1/poly(|x|)$, it must be that $z = f(x)$.
\end{enumerate}
\end{definition}

As mentioned in the introduction, we add an additional requirement to our QGES known as offline-ness. What this means is that the client can send a quantum state to the server, representing an encryption of the input, and decide afterwards which input it intends to use. Essentially the client is free to change its mind about the input and not commit to a particular input when the protocol commences. More formally, we define the offline condition as follows:
\begin{definition} 
[Offline QGES] \label{defn:offline}
In a QGES, let:
\begin{equation}
\Ket{\psi_x} = \frac{1}{\sqrt{|\mathcal{K}_C(x)|}} \sum\limits_{k^x_i \in \mathcal{K}_C(x)} \Ket{k^x_i}_K \Ket{y^x_i}_E
\end{equation}
be the state of the client's system on input $x$ in a superposition over possible keys and encryptions.
Here, $\mathcal{K}_C(x)$ is the set of encryption keys which are compatible with $x$ (i.e.\ could have resulted from the key generation algorithm acting on $x$, $K(x)$) and $\Ket{y^x_i}_E \leftarrow QE(x, k^x_i)$ is the quantum encryption of $x$ using the key $k^x_i$. The subscripts $K$ and $E$ indicate the key register and the encrypted input register, respectively.

For $n > 0$, let $x_1, x_2 \in \{0, 1\}^n$ be two different inputs for $f$ such that $x_1 \neq x_2$.
We say that the QGES is offline if there exists a polynomial-sized quantum circuit, $V$, acting only on the key register, $K$, such that $V \ket{\psi_{x_1}} = \ket{\psi_{x_1}}$.
\end{definition}
The definition essentially says that there exists an efficient procedure that the client can apply on her local system in order to map from one input of size $n$, to a different input of the same size.
One might ask whether this property is not immediately satisfied by a QGES leaking only the size of the input. Indeed, leaking only the size of the input is equivalent to saying that the density matrix corresponding to the quantum encryption, which the server receives, is the same for all inputs of the same size. In other words, it is the case that:
\begin{equation}
Tr_{K}(\ket{\psi_{x_1}}\bra{\psi_{x_1}}) = Tr_{K}(\ket{\psi_{x_2}}\bra{\psi_{x_2}})
\end{equation}
for any $x_1, x_2 \in \{0, 1\}^n$.
Since the density matrix is the same, that means that there exists a unitary (acting on the client's system) which can map one purification of this state, corresponding to one input, to another, corresponding to a different input. This unitary must be verifiable in the sense that the client can check (using a quantum \textsf{SWAP} test) whether the unitary maps to the correct purification. However, this unitary \emph{need not have a polynomial-sized quantum circuit representation}. 
The offline condition simply imposes that such a circuit always exists.
As we've mentioned before, UBQC and all other existing delegated blind quantum computation protocols satisfy this property.

We are now ready to prove Theorem~\ref{thm:qges}, showing that functions that can be delegated by the client to the server in an offline QGES are contained in $\qgesdec$:
\begin{proof}[Proof of Theorem~\ref{thm:qges}]
For an input $x$ for which the client wants to compute $f(x)$, consider the state: \\
\begin{equation}
\Ket{\psi_x} = \frac{1}{\sqrt{|\mathcal{K}_C(x)|}} \sum\limits_{k^x_i \in \mathcal{K}_C(x)} \Ket{k^x_i}_K \Ket{y^x_i}_E
\end{equation}
where the notation is the same as in Definition~\ref{defn:offline}.
If we trace out the key register, $K$, the resulting density matrix is the mixed state of possible encrypted states which the server will receive:
\begin{equation}
\rho_x = \frac{1}{|\mathcal{K}_C(x)|} \sum\limits_{k^x_i \in \mathcal{K}_C(x)} \Ket{y^x_i}\Bra{y^x_i}
\end{equation}
The assumption that the protocol only leaks the size of the input $x$ to the server implies that for any two inputs $x_1$, $x_2$ it is the case that $\rho_{x_1} = \rho_{x_2}$. In fact, something stronger is true. 
The server's `view' of the protocol should be the same in both cases. By this we mean that the server's system as well as the distribution of his classical messages must be independent of $x$, given the size of $x$. Therefore, we should consider a state comprising of his system and his response after receiving the quantum encryption, for a particular input, $x$:
\begin{equation}
\Ket{\phi_x} = \frac{1}{\sqrt{|\mathcal{K}_C(x)|}} \sum\limits_{k^x_i \in \mathcal{K}_C(x)}
 \Ket{k^x_i}_K \; U_{ERS} \Ket{y^x_i}_E \Ket{0}^{\otimes t}_{R} \Ket{anc}_{S}
\end{equation}
Here, $U_{ERS}$ is the unitary performed by the server in order to produce his response, which will be stored in the response register, initially set to $\Ket{0}^{\otimes t}_{R}$, where $t = poly(|x|)$. This unitary will of course involve the encrypted state provided by the client and the server's private ancilla, denoted as $\Ket{anc}_S$ (but will not involve the key register).
Note that the server's response is a classical bit string. Hence, the state in the response register, obtained through the application of the unitary $U_{ERS}$, will be a probabilistic mixture over computational basis states. This, however, makes no difference in our proof and we can just as well assume that his response is a general quantum state.

If we again trace out the register $K$ we obtain some state $\sigma_x = Tr_{K}(\Ket{\phi_x}\Bra{\phi_x})$. This state encodes the distribution of possible messages exchanged by the client and the server in one round of interaction, as well as the server's private system. Since $\rho_{x_1} = \rho_{x_2}$, it is also the case that $\sigma_{x_1} = \sigma_{x_2}$. By the purification principle, this means that if we consider the states:
\begin{equation}
\Ket{\phi_{x_1}} = \frac{1}{\sqrt{|\mathcal{K}_C(x_1)|}} \sum\limits_{k^{x_1}_i \in \mathcal{K}_C(x_1)}
 \Ket{k^{x_1}_i}_K \; U_{ERS} \Ket{y^{x_1}_i}_E \Ket{0}^{\otimes t}_{R} \Ket{anc}_{S}
\end{equation}
\begin{equation}
\Ket{\phi_{x_2}} = \frac{1}{\sqrt{|\mathcal{K}_C(x_2)|}} \sum\limits_{k^{x_2}_i \in \mathcal{K}_C(x_2)}
 \Ket{k^{x_2}_i}_K \; U_{ERS} \Ket{y^{x_2}_i}_E \Ket{0}^{\otimes t}_{R} \Ket{anc}_{S}
\end{equation}
there exists a local unitary, $V_K$, acting only on the key register which can map $\Ket{\phi_{x_1}}$ to $\Ket{\phi_{x_2}}$, for any two inputs $x_1$ and $x_2$. In fact, let us examine the states of the system for inputs $x_1$ and $x_2$ before $U_{ERS}$ is applied:
\begin{equation}
\Ket{\chi_{x_1}} = \frac{1}{\sqrt{|\mathcal{K}_C(x_1)|}} \sum\limits_{k^{x_1}_i \in \mathcal{K}_C(x_1)}
 \Ket{k^{x_1}_i}_K \; \Ket{y^{x_1}_i}_E \Ket{0}^{\otimes t}_{R} \Ket{anc}_{S}
= \Ket{\psi_{x_1}}_{KE} \Ket{0}^{\otimes t}_{R} \Ket{anc}_{S}
\end{equation}
\begin{equation}
\Ket{\chi_{x_2}} = \frac{1}{\sqrt{|\mathcal{K}_C(x_2)|}} \sum\limits_{k^{x_2}_i \in \mathcal{K}_C(x_2)}
 \Ket{k^{x_2}_i}_K \; \Ket{y^{x_2}_i}_E \Ket{0}^{\otimes t}_{R} \Ket{anc}_{S}
 = \Ket{\psi_{x_2}}_{KE} \Ket{0}^{\otimes t}_{R} \Ket{anc}_{S}
\end{equation}
These states are also related by $V_K$. This can be inferred from the following relations.
First, we know that:
\begin{equation}
(V_K \otimes I_{ERS}) \Ket{\phi_{x_1}} = \Ket{\phi_{x_2}}
\end{equation}
And also that:
\begin{equation}
(I_K \otimes U_{ERS}) \Ket{\chi_{x_1}} = \Ket{\phi_{x_1}} \; \; \; \; \; \;
(I_K \otimes U_{ERS}) \Ket{\chi_{x_2}} = \Ket{\phi_{x_2}}
\end{equation}
Therefore:
\begin{equation}
(I_K \otimes U^{\dagger}_{ERS})(V_K \otimes I_{ERS})(I_K \otimes U_{ERS}) \Ket{\chi_{x_1}} = \Ket{\chi_{x_2}}
\end{equation}
But $(I_K \otimes U^{\dagger}_{ERS})$ and $V_K \otimes I_{ERS}$ commute because they act on different systems and therefore $(I_K \otimes U^{\dagger}_{ERS})$ and $(I_K \otimes U_{ERS})$ will cancel out, leaving:
\begin{equation}
(V_K \otimes I_{ERS}) \Ket{\chi_{x_1}} = \Ket{\chi_{x_2}}
\end{equation}
This is also illustrated in the following diagram:
\begin{center}
\scalebox{1.5}{
\begin{xy}
\xymatrix@C=8em@R=4em{
   \ket{\chi_{x_1}} \ar@{-->}[r]^{V_K \otimes I_{ERS}} \ar@<-4pt>[d]_{I_K \otimes U_{ERS}} & \ket{\chi_{x_2}}  \\
   \ket{\phi_{x_1}} \ar[r]^{V_K \otimes I_{ERS}} & \ket{\phi_{x_2}}  \ar@<-2pt>[u]_{I_K \otimes U^{\dagger}_{ERS}}
}
\end{xy}
}
\end{center}
Since the protocol is offline, we know that $V_K$ must be a polynomial-sized quantum circuit.
Note that even if we trace out the server's ancilla from the states $\Ket{\phi_{x_1}}$ and $\Ket{\phi_{x_2}}$, the resulting states are still related by $V_K$ on the key register.
This allows us to define a $\cQCMAqpoly$ algorithm computing any function which admits a QGES.
To do so we first introduce some notation. We will consider the following states:
\begin{equation}
\Ket{\kappa_x} = \frac{1}{\sqrt{|\mathcal{K}_C(x)|}} \sum\limits_{k^x_i \in \mathcal{K}_C(x)} \Ket{k^x_i}_K
\end{equation}
\begin{equation}
\Ket{\kappa_{x'}} = \frac{1}{\sqrt{|\mathcal{K}_C(x')|}} \sum\limits_{k^{x'}_i \in \mathcal{K}_C(x')} \Ket{k^{x'}_i}_K
\end{equation}
which are simply superpositions over the valid keys for two different inputs $x$ and $x'$. Next we consider:
\begin{equation}
\Ket{\phi_{x}} = \frac{1}{\sqrt{|\mathcal{K}_C(x)|}} \sum\limits_{k^{x}_i \in \mathcal{K}_C(x)}
 \Ket{k^{x}_i}_K \; U_{ERS} \Ket{y^{x}_i}_E \Ket{0}^{\otimes t}_{R} \Ket{anc}_{S}
\end{equation}
\begin{equation}
\Ket{\phi_{x'}} = \frac{1}{\sqrt{|\mathcal{K}_C(x')|}} \sum\limits_{k^{x'}_i \in \mathcal{K}_C(x')}
 \Ket{k^{x'}_i}_K \; U_{ERS} \Ket{y^{x'}_i}_E \Ket{0}^{\otimes t}_{R} \Ket{anc}_{S}
\end{equation}
which include the encrypted states and the server's response. Lastly, we trace out the server's ancilla from both these states resulting in:
\begin{equation}
\omega_x = Tr_{S}(\Ket{\phi_{x}}\Bra{\phi_{x}})
\end{equation}
\begin{equation}
\omega_{x'} = Tr_{S}(\Ket{\phi_{x'}}\Bra{\phi_{x'}})
\end{equation}
From the above argument the two states $\Ket{\kappa_x}$ and $\Ket{\kappa_{x'}}$ and the two states $\omega_x$ and $\omega_{x'}$ are related through the same polynomial-sized quantum circuit $V_K$ acting only on the key register.

We can now present the algorithm. Let us first consider the one round case. The algorithm would work as follows:
\begin{enumerate}
\item The input to the algorithm is some string $x$ for which we want to compute $f(x)$.
\item The algorithm receives as advice the string $x'$ which is simply some string of the same length as $x$. Additionally,
it receives the state $\omega_{x'}$.
It is clear that both of these only depend on $|x|$ and have a length which is polynomial in $|x|$ hence constituting a valid advice.
\item From the definition of the key generating function, the algorithm can efficiently produce the states $\Ket{\kappa_x}$ and $\Ket{\kappa_{x'}}$.
\item The classical witness is a description of the quantum circuit $V^{\dagger}_K$.
\item The algorithm tests that $V^{\dagger}_K$ maps $\Ket{\kappa_{x'}}$ to $\Ket{\kappa_{x}}$. This can be done through a quantum \textsf{SWAP} test.
\item Use $V^{\dagger}_K$ to map $\omega_{x'}$ to $\omega_x$.
\item By measuring the response register of $\omega_x$, the algorithm obtains the response that the server would have produced in an interaction with the client in the QGES protocol. Applying the decryption algorithm to this response will yield the correct result $f(x)$ with high probability.
\end{enumerate}
The probability of success of the algorithm can be boosted by providing polynomially many copies of $\omega_{x'}$ as advice and performing multiple \textsf{SWAP} tests. Additionally this algorithm can be made to compute the complement of $f(x)$ as well which would gives us a $\ccoQCMAqpoly$ containment.

Let us now consider the case of polynomially many rounds of interaction, of which the first round involves quantum interaction. We've shown that for that first round, we can design a polynomial-time quantum algorithm, receiving quantum advice and a classical witness, to emulate the client in the QGES. For the rounds after the quantum interaction, we essentially have a classical GES. 
Our algorithm would then operate as in the proof of Theorem~\ref{thm:afk}. In other words, apart from the quantum advice and the witness providing the circuit $V_K$, the algorithm will also receive classical (randomised) advice corresponding to a transcript of the classical interaction in the protocol, and a key $k$ to decrypt this transcript. It follows that $f$ will be contained in $\qgesdec$.
\end{proof}

A question we might ask is: why can't we adapt this proof to the case of full quantum communication between the client and the server? For this case, we would need a stronger version of the offline-ness condition that imposes the existence of a polynomial-size quantum circuit for mapping \emph{the entire transcript} of the protocol, corresponding to some input $x_1$ to that of some input $x_2$ (where $|x_1| = |x_2|$).
Another question is whether the class of interest should in fact be $\cBQP^{\qgesdec}$ since a $\cBQP$ client could use the QGES as an oracle. Just as we showed that $\cBPP^{\gesdec} = \gesdec$ we can show that:
\begin{lemma}
$\cBQP^{\qgesdec} = \qgesdec$
\end{lemma}
\begin{proof}
This proof is similar to the one showing that $\cBPP^{\gesdec} = \gesdec$. Just like in that case, the inclusion
$\qgesdec \subseteq \cBQP^{\qgesdec}$ is immediate and we need only show that $\cBQP^{\qgesdec} \subseteq \cQCMAqpoly$. The containment in $\ccoQCMAqpoly$ follows by complementation.

Consider a quantum algorithm \emph{QA} for deciding problems in $\cBQP^{\qgesdec}$. We will show that this algorithm can be simulated by a $\cQCMAqpoly$ algorithm, denoted \emph{NQA}. Since $\cBQP$, $\cQCMA$ and $\ccoQCMA$ have bounded error in deciding problems, we can assume, from standard amplification techniques, that this error is of order $2^{-poly(n)}$, where $n$ is the size of the input. We will also assume that for all quantum algorithms measurements are postponed until the end of the circuit.

We will treat the case without advice first, and then explain how to deal with the quantum advice at the end.
To start with, \emph{NQA} will simulate \emph{QA} until it makes a query to the oracle. In the standard definition of oracles the oracle is just a classical function that solves a decision problem. However, when dealing with quantum algorithms such as \emph{QA} it is also possible to speak of quantum oracles, where the oracle can be viewed as some unitary operation (technically a sequence of unitary operations for each possible input length, see \cite{ak} for more details) which \emph{QA} can query even in superposition. Our result will cover this more general case of quantum oracles. We would therefore like the \emph{NQA} algorithm to be able to simulate this quantum oracle.

Firstly, just like in the classical case we have that if some language $L \in \cQCMA \cap \ccoQCMA$ then $L \in \cQCMA$ and $L^c \in \ccoQCMA$, where $L^c$ is the complement of $L$. This means that there exist polynomial-sized quantum circuits $Q_L$ and $Q_{L^c}$ which take some input $x$ along with classical witnesses $w_1$ and $w_2$, respectively, and decide correctly, when the output is measured, with probability at least, $1 - 2^{-poly(|x|)}$.
In other words, $Q_L$ receives as input $\Ket{x}\Ket{w_1}\ket{0^m}$ and $Q_{L^c}$ receives as input $\ket{x}\ket{w_2}\ket{0^m}$, respectively, where $m = poly(|x|)$.
If we were to run both $Q_L$ and $Q_{L^c}$ on $x$, because $L$ and $L^c$ are complementary, the output qubits, when measured, will also be complementary with high probability.

Assume that $Q_L$ and $Q_{L^c}$ are circuits which act on $t = poly(|x|)$ qubits.
We define a new quantum circuit called $SimQuery$ which operates on $2t + 1$ qubits. $SimQuery$ applies $Q_L$ to the first $t$ qubits and $Q_{L^c}$ to the next $t$ qubits. It then applies a Pauli $\mathsf{X}$ to the output qubit of $Q_{L^c}$ and a $\textsf{CCNOT}$ operation from the output qubits of $Q_L$ and $Q_{L^c}$ onto the the $2t+1$'th qubit.
It then applies $\mathsf{X}$ again to the output qubit of $Q_{L^c}$ and then $Q_L^{\dagger}$ and $Q_{L^c}^{\dagger}$ on the first $2m$ qubits. An illustration of this circuit (acting on the $\ket{00...0}$ input) is given in Figure~\ref{fig:simquery}.
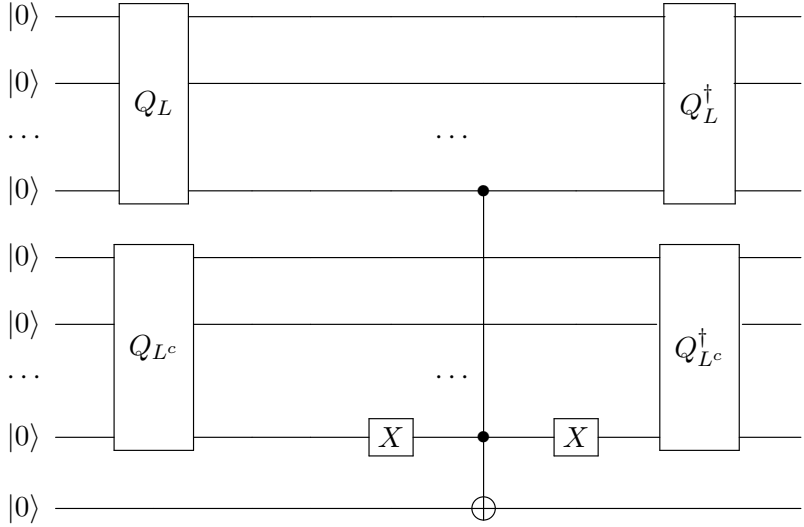
\begin{figure}[ht!]
\[
\Qcircuit @C=2em @R=1.4em {
	& \lstick{\Ket{0}} & \multigate{3}{Q_L} & \qw & \qw & \qw & \qw & \qw & \multigate{3}{Q_L^{\dagger}} & \qw \\
	& \lstick{\Ket{0}} & \ghost{Q_L} & \qw & \qw & \qw & \qw & \qw & \ghost{Q_L^{\dagger}} & \qw \\
	& \lstick{\cdots}  &  &  &  & & \lstick{\cdots} &  &  &  \\
	& \lstick{\Ket{0}} & \ghost{Q_L} & \qw & \qw & \qw & \ctrl{4} & \qw & \ghost{Q_L^+} & \qw \\		
	& \lstick{\Ket{0}} & \multigate{3}{Q_{L^c}} & \qw & \qw & \qw & \qw & \qw & \multigate{3}{Q_{L^c}^{\dagger}} & \qw \\
	& \lstick{\Ket{0}} & \ghost{Q_{L^c}} & \qw & \qw & \qw & \qw & \qw & \ghost{Q_{L^C}^+} & \qw \\
	& \lstick{\cdots}  &  &  &  &  & \lstick{\cdots} &  &  &  \\
	& \lstick{\Ket{0}} & \ghost{Q_{L^c}} & \qw & \qw & \gate{X} & \ctrl{1} & \gate{X} & \ghost{Q_{L^c}^{\dagger}} & \qw \\
	& \lstick{\Ket{0}} & \qw & \qw & \qw & \qw & \targ & \qw & \qw & \qw
}
\]
\caption{Quantum circuit for $SimQuery$ acting on the $\ket{00...0}$ input.}
\label{fig:simquery}
\end{figure}

The $\textsf{CCNOT}$ operation flips its target qubit if the control qubits are in the state $\ket{11}$. The effect of the first Pauli $\mathsf{X}$ is to flip the outcome when the control qubits are in the state $\ket{10}$. Roughly speaking, $SimQuery$ will flip the final qubit if $Q_L$ accepts and $Q_{L^c}$ rejects. The reason for then applying the two circuits in reverse is to `uncompute' their result and only leave the $2t+1$ qubit flipped whenever $Q_L$ accepts and $Q_{L^c}$ rejects.

We can now explain how \emph{NQA} can use $SimQuery$ to simulate a query of \emph{QA}. Suppose \emph{QA} queries the oracle for some input $x$ testing to see if it is in $L$, for some $L \in \cQCMA \cap \ccoQCMA$. In other words, $\ket{x}\Ket{0} \rightarrow \Ket{x}\ket{1}$, with high probability, if $x \in L$ and $\ket{x}\Ket{0} \rightarrow \Ket{x}\ket{0}$, with high probability, if $x \not\in L$. \emph{NQA} will then run $SimQuery$ with input $\ket{x}\ket{w_1}\ket{0^m}\ket{x}\ket{w_2}\ket{0^m}\ket{0}$, where $w_1$ and $w_2$ are the witnesses from before and $m = poly(|x|)$.
The effect of this will be to flip the final qubit if $x \in L$ and leave it unchanged if $x \not\in L$, with high probability. This is true because of the complementarity of $Q_L$ and $Q_{L^c}$ (when one accepts the other rejects and viceversa, except with small probability).

This procedure will simulate the query that \emph{QA} performs. \emph{NQA} then uses the last qubit from $SimQuery$ as the query response qubit and continues to do this for all other queries of \emph{QA} and otherwise simulate \emph{QA} exactly. Note that each simulated query has some small probability of not matching the actual query of \emph{QA} when a measurement is performed. However, as mentioned, this probability is exponentially small. Since there are polynomially many queries in total, by a union bound, the probability that at least one simulated query behaves incorrectly will still be exponentially small.

Adding quantum advice to this picture does not change much. Just like in the classical case, we can assume that \emph{NQA} receives as advice a concatenation of all advice states used by the oracle of \emph{QA}. The quantum circuits $Q_L$, $Q_{L^{c}}$ and $SimQuery$ are then extended with polynomially many qubits to act on this advice as well.

It is therefore the case that $\cBQP^{\qgesdec} \subseteq \qgesdec$ and our result follows directly.
\end{proof}

\noindent Note that in the QGES, the client need not be a fully $\cBQP$-capable machine (and indeed, for UBQC the only quantum capabilities of the client are to prepare single qubits).

What happens if we drop the ``offline'' requirement for this scheme? As mentioned, that would imply that the mapping from one purification of the encrypted quantum state to another can in principle be any unitary operation so long as the client can check that the correct mapping was performed (with high probability). Having such a weak restriction on this unitary makes it very difficult to impose an upper bound on the types of computations that are allowed by such a scheme. Indeed, the offline-ness condition plays a crucial role in our proof of Theorem~\ref{thm:qges}. At the same time, it is arguably a very natural condition to have in any realistic protocol. We therefore leave the online case as an open problem.

\subsection{QGES and $\textsf{NP}$-hard functions}
Theorem~\ref{thm:qges} can be viewed as a quantum version of Theorem~\ref{thm:afk} which, as mentioned,
was used by Abadi et al.\ to show that there can be no GES for $\cNP$-hard functions unless the polynomial hierarchy collapses.
As we have stated before, for quantum computers, the possibility of delegating $\cNP$-complete problems makes, arguably, even more sense since Grover's algorithm offers a quadratic speed-up in solving such problems \cite{grover}. Alas, we show that even with a QGES, delegating such problems seems unlikely.

Since we have shown that functions which admit an offline QGES are contained in $\qgesdec$, and since if $\cNP \subset \qgesdec$ then $\ccoNP \subset \qgesdec$, to prove Theorem \ref{thm:QGESNP}, it suffices to show that if $\ccoNP \subset \cQCMAqpoly$  then, informally speaking, the polynomial hierarchy \textquotedblleft
comes about as close to collapsing as one could reasonably hope to prove given
a quantum hypothesis\textquotedblright---and more specifically, that $\Pi
_{3}^{\mathsf{P}}\subseteq\left(  \Sigma_{2}^{\mathsf{P}}\right)
^{\mathsf{P{}romiseQMA}}$. \ Here a $\mathsf{P{}romiseQMA}$\ oracle means an
oracle for some $\mathsf{P{}romiseQMA}$-complete promise problem $\left(
\Pi_{\operatorname*{YES}},\Pi_{\operatorname*{NO}}\right)  $, whose responses
can be arbitrary on inputs $x\notin\Pi_{\operatorname*{YES}}\cup
\Pi_{\operatorname*{NO}}$\ that violate the promise. \ We don't even demand
that the oracle's responses, on promise-violating inputs, be consistent from
one query to the next. \ On the other hand, it does need to be possible to
query the $\mathsf{P{}romiseQMA}$\ oracle on some promise-violating inputs,
without such queries causing the entire algorithm to abort.

The starting point for all such collapse results, of course, is the
Karp-Lipton Theorem \cite{kp}, which says that if $\mathsf{NP}\subset
\mathsf{P/poly}$\ then $\Pi_{2}^{\mathsf{P}}\subseteq\Sigma_{2}^{\mathsf{P}}$,
and hence the polynomial hierarchy collapses to the second level. \ An easy
extension of the Karp-Lipton theorem, proved by Yap \cite{yap}, which we now reprove for completeness,
shows that if $\mathsf{coNP}\subset\mathsf{NP/poly}$, then $\mathsf{PH}%
$\ collapses to the \textit{third} level.

\begin{proposition}
\label{klext}If $\mathsf{coNP}\subset\mathsf{NP/poly}$, then $\Pi
_{3}^{\mathsf{P}}\subseteq\Sigma_{3}^{\mathsf{P}}$.
\end{proposition}

\begin{proof}
Abusing notation, here and later in this section we'll use $\Phi$, $\Psi$,
etc.\ to refer not only to SAT\footnote{SAT stands for \emph{boolean satisfiability} and a SAT instance, in our context, simply refers to a boolean formula. The problem of deciding whether a given boolean formula admits and assignment of variables that evaluates to true (i.e. a satisfying assignment) is $\cNP$-complete.} instances but to strings encoding those
instances. \ Also, if (say) $\Phi\left(  x,y,z\right)  $ is a SAT instance
taking multiple strings as input, then by $\Phi\left(  x,y\right)  $, we'll
mean the instance obtained from $\Phi$ by fixing the variables in $x$ and $y$,
and leaving only the variables in $z$ as free variables.

A $\Pi_{3}^{\mathsf{P}}$ sentence has the form%
\[
S=\text{\textquotedblleft}\forall x\exists y\forall z~\Phi\left(
x,y,z\right)  \text{\textquotedblright}%
\]
where $x,y,z$\ are strings of some given polynomial size, and $\Phi$\ is a
polynomial-time computable predicate (without loss of generality, a SAT
instance). \ Under the stated hypothesis, we need to show how to decide $S$ in
$\Sigma_{3}^{\mathsf{P}}$.

Let $C$ be the assumed $\mathsf{NP/poly}$\ algorithm for $\mathsf{coNP}$, and
let $a$ be its advice. \ Then by hypothesis, for all SAT instances $\Psi$, if
$\Psi$\ is unsatisfiable then there exists a witness $w$ such that $C\left(
\Psi,w,a\right)  $\ accepts, while if $\Psi$\ is satisfiable then $C\left(
\Psi,w,a\right)  $\ rejects for all $w$.

Our $\Sigma_{3}^{\mathsf{P}}$\ rewriting of $S$ is now as follows:\bigskip

\textbf{There exists an advice string }$a$\textbf{, such that}

\begin{enumerate}
\item[(1)] \textbf{(Completeness of }$\mathbf{C}$\textbf{)} For all SAT
instances $\Psi$, either there exists a $z$ that satisfies $\Psi$, or else
there exists a $w$ such that $C\left(  \Psi,w,a\right)  $\ accepts.

\item[(2)] \textbf{(Soundness of }$\mathbf{C}$\textbf{)} For all SAT instances
$\Psi$, all satisfying assignments $z$\ for $\Psi$, and all $w$, the procedure
$C\left(  \Psi,w,a\right)  $\ rejects.

\item[(3)] \textbf{(Truth of }$\mathbf{S}$\textbf{)} For all $x$, there exists
a $y$ as well as a witness $w$ such that $C\left(  \urcorner\Phi\left(
x,y\right)  ,w,a\right)  $\ accepts. \ (In other words, there is no $z$ that
makes $\Phi\left(  x,y,z\right)  $\ false.)
\end{enumerate}
\end{proof}

Proposition \ref{klext} is what we seek to imitate in the quantum setting,
getting whatever leverage we can from the weaker assumption $\mathsf{coNP}%
\subset\mathsf{QCMA/qpoly}$.

Note that, had we assumed (say) $\mathsf{coNP}\subset\mathsf{QCMA/poly}$, it
would be routine to mimic the usual Karp-Lipton argument, merely substituting
the class $\mathsf{P{}romiseQCMA}$ for $\mathsf{NP}$\ at appropriate points in
the proof of Proposition \ref{klext}. \ This would give us the collapse
$\Pi_{3}^{\mathsf{P}}\subseteq\left(  \Sigma_{2}^{\mathsf{P}}\right)
^{\mathsf{P{}romiseQCMA}}$. \ However, the fundamental difficulty we face is
that our hypothesised nonuniform algorithm uses \textit{quantum advice
states}. \ And while a $\mathsf{P{}romiseQMA}$\ machine can simply guess a
quantum advice state $\sigma$, it can't then pass $\sigma$\ to an oracle, at
least not with conventional oracle calls. \ (To allow the passing of quantum
states to oracles, we would need \textit{quantum oracles}, as studied for
example by Aaronson and Kuperberg \cite{ak}.)

To get around this difficulty, we'll rely essentially on a 2010 result of
Aaronson and Drucker \cite{qkp, aaronsondrucker},\ characterising the power of
quantum advice. \ These authors proved that $\mathsf{BQP/qpoly}$\ is contained
in $\mathsf{QMA/poly}$---and even more strongly,

\begin{theorem}
\label{adthm}
$\mathsf{BQP/qpoly=YQP/poly}$.
\end{theorem}

Here $\mathsf{YQP}$, known as Yoda quantum polynomial-time, is the class of problems solvable by a polynomial-time
quantum algorithm with help from a polynomial-size \textit{untrusted} quantum
advice state that depends only on the input length $n$. \ In other words,
Theorem \ref{adthm}\ says that we can simulate trusted quantum advice by
trusted classical advice combined with untrusted quantum advice, by using the
classical advice to verify the quantum advice for ourselves.

By using Theorem \ref{adthm}, to replace a quantification over quantum advice
states by a quantification over classical advice strings, Aaronson and Drucker were able to show the following:

\begin{theorem}
\label{adcollapse}If $\mathsf{NP}\subset\mathsf{BQP/qpoly}$, $\Pi
_{2}^{\mathsf{P}}\subseteq\mathsf{QMA}^{\mathsf{P{}romiseQMA}}$.
\end{theorem}

By adapting our argument from later in this section, one can actually improve
Theorem \ref{adcollapse}, to show that if $\mathsf{NP}\subset
\mathsf{BQP/qpoly}$ then $\Pi_{2}^{\mathsf{P}}\subseteq\mathsf{NP}%
^{\mathsf{P{}romiseQMA}}$. \ In any case, we now seek a common generalisation
of the proofs of Proposition \ref{klext} and Theorem \ref{adcollapse}, to get
a collapse from the assumption $\mathsf{coNP}\subset\mathsf{QCMA/qpoly}$.

As Aaronson and Drucker \cite{aaronsondrucker}\ pointed out, a simple
extension of their proof of Theorem \ref{adthm}\ gives $\mathsf{QCMA/qpoly}%
\subseteq\mathsf{QMA/poly}$, and even the following.

\begin{theorem}
\label{adthm2}%
$\mathsf{QCMA/qpoly}=\mathsf{YQ\cdot QCMA/poly}.$
\end{theorem}

Here the $\mathsf{YQ\cdot}$ operator simply adds untrusted quantum advice to
whatever (quantum) complexity class it acts on. \ Thus $\mathsf{YQ\cdot
BQP}=\mathsf{YQP}$, while for completeness:

\begin{definition}
$\mathsf{YQ\cdot QCMA}$ is the class of languages $L$ for which there exist
polynomial-time quantum algorithms $C$ and $V$, such that for all input
lengths $n$:

\begin{itemize}
\item There exists a polynomial-size quantum advice state $\sigma_{n}$ such
that $V\left(  0^{n},\sigma_{n}\right)  $\ accepts with probability at least
$0.99$. \ If $V\left(  0^{n},\sigma\right)  $\ accepts with probability at
least $0.98$, then we call the advice state $\sigma$\ \textquotedblleft
valid\textquotedblright\ for input length $n$.

\item For all inputs $x\in\left\{  0,1\right\}  ^{n}\cap L$ and all valid
$\sigma$, there exists a polynomial-size classical witness $w$ such that
$C\left(  x,w,\sigma\right)  $ accepts with probability at least $2/3$\ .

\item For all inputs $x\in\left\{  0,1\right\}  ^{n}\setminus L$, all
classical witnesses $w$, and all valid $\sigma$, we have that $C\left(
x,w,\sigma\right)  $ accepts with probability at most $1/3$.
\end{itemize}
\end{definition}

In what follows, we'll need one additional observation about the proof of
Theorem \ref{adthm2}. \ Namely, in our $\mathsf{YQ\cdot QCMA/poly}%
$\ simulation of $\mathsf{QCMA/qpoly}$, without loss of generality we can
choose the classical advice string $a=a_{n}$ in such a way that there's
essentially just \textit{one} valid quantum advice state\ compatible with $a$.
\ Or more precisely: we can ensure that, for all $\rho_{1},\rho_{2}$\ such
that $V\left(  0^{n},a,\rho_{1}\right)  $\ and $V\left(  0^{n},a,\rho
_{2}\right)  $\ both accept with probability at least $0.98$, and all $x$ and
$w$, we have (say)%
\[
\left\vert \Pr\left[  C\left(  x,w,a,\rho_{1}\right)  \text{ accepts}\right]
-\Pr\left[  C\left(  x,w,a,\rho_{2}\right)  \text{ accepts}\right]
\right\vert <\frac{1}{20}.
\]
This is because Theorem \ref{adthm2}, like Theorem \ref{adthm}, is proven via
the method of \textquotedblleft majority-certificates,\textquotedblright\ in
which given a polynomial-time quantum algorithm $Q$, one verifies that an
unknown quantum state $\rho$\ leads to approximately the desired values of
$\Pr\left[  Q\left(  x,\rho\right)  \text{ accepts}\right]  $\ for
\textit{each} of exponentially many different inputs $x$, via a measurement of
$\rho$\ that takes only polynomial time. \ We note that this works only
because of special structure in $\rho$---but for any state $\sigma$, there
exists another state $\rho$\ that has the requisite special structure, as well
as a modified quantum algorithm $Q^{\prime}$, such that%
\[
\Pr\left[  Q^{\prime}\left(  x,\rho\right)  \text{ accepts}\right]  \approx
\Pr\left[  Q\left(  x,\sigma\right)  \text{ accepts}\right]
\]
for all $x$.

We're finally ready to prove Theorem \ref{thm:QGESNP}.
\begin{proof}[Proof of Theorem~\ref{thm:QGESNP}]
Essentially, we are going to show that if $\mathsf{coNP}\subset\mathsf{QCMA/qpoly}$, then $\Pi_{3}%
^{\mathsf{P}}\subseteq\left(  \Sigma_{2}^{\mathsf{P}}\right)  ^{\mathsf{P{}%
romiseQMA}}$.
A $\Pi_{3}^{\mathsf{P}}$ sentence has the form%
\[
S=\text{\textquotedblleft}\forall x\exists y\forall z~\Phi\left(
x,y,z\right)  \text{\textquotedblright}%
\]
where $x,y,z$\ are strings of some given polynomial size, and $\Phi$\ is a
polynomial-time computable predicate. \ Under the stated hypothesis, we need
to show how to decide $S$ in $\mathsf{NP}^{\mathsf{NP}^{\mathsf{P{}romiseQMA}%
}}$.

By Theorem \ref{adthm2}, the hypothesis $\mathsf{coNP\subset QCMA/qpoly}$\ is
equivalent to $\mathsf{coNP\subset YQ\cdot QCMA/poly}$. \ In other words: we
can assume that there exists a polynomial-time quantum algorithm $C\left(
\Phi,w,a,\sigma\right)  $, which takes as input a SAT instance $\Phi$, a
classical witness $w$, a classical advice string $a$, and a quantum advice
state $\sigma$. \ Assuming $a$ and $\sigma$ are the correct $\mathsf{YQ\cdot
QCMA/poly}$ advice, $C$ checks whether $w$ is a witness to $\Phi$'s
\textit{un}satisfiability. \ This is a sound and complete proof system for
$\mathsf{coNP}$, in the sense that, again assuming the correctness of $a$ and
$\sigma$,

\begin{enumerate}
\item[(i)] for every unsatisfiable $\Phi$, there exists a $w$ such that
$C\left(  \Phi,w,a,\sigma\right)  $\ accepts with probability at least $2/3$,

\item[(ii)] for no satisfiable $\Phi$ does there exist a $w$ such that
$C\left(  \Phi,w,a,\sigma\right)  $\ accepts with probability more than $1/3$.
\end{enumerate}

Moreover, as discussed above, there exists an $a$ such that the state $\sigma
$\ is essentially unique, in the sense that%
\[
\Pr\left[  C\left(  \Psi,w,a,\rho_{1}\right)  \text{ accepts}\right]
\approx\Pr\left[  C\left(  \Psi,w,a,\rho_{2}\right)  \text{ accepts}\right]
\]
for all valid $\rho_{1},\rho_{2}$.

Our job is to rewrite $S$ as an $\mathsf{NP}^{\mathsf{NP}^{\mathsf{P{}%
romiseQMA}}}$ sentence. \ Our rewriting will be as follows:\bigskip

\textbf{There exists a classical advice string }$a$\textbf{ such that}

\begin{enumerate}
\item[(1)] for all valid quantum advice states $\rho_{1},\rho_{2}$, all SAT
instances $\Psi$, and all assignments $w$, we have%
\[
\left\vert \Pr\left[  C\left(  \Psi,w,a,\rho_{1}\right)  \text{ accepts}%
\right]  -\Pr\left[  C\left(  \Psi,w,a,\rho_{2}\right)  \text{ accepts}%
\right]  \right\vert <\frac{1}{10}.
\]

(In words: the classical advice string $a$ uniquely determines the behaviour of
$C$, once we find a valid quantum advice state $\sigma$ that's compatible with
$a$.)

\item[(2)] For all SAT instances $\Psi$, there exists a valid quantum advice
state $\sigma$, as well as either an assignment $z$ that satisfies $\Psi$, or
else a classical witness $w$ such that $C\left(  \Psi,w,a,\sigma\right)  $
accepts with probability at least $2/3$.

(In words: the advice $a$ leads to a complete procedure for deciding the class
$\mathsf{coNP}$, and specifically the UNSAT problem, in $\mathsf{YQ\cdot
QCMA/poly}$. \ That is, once we find a valid advice state $\sigma$, the
quantum algorithm $C$ then accepts every SAT instance $\Psi$ that has no
satisfying assignment.)

\item[(3)] For all valid quantum advice states $\sigma$, all SAT instances
$\Psi$, all $z$, and all $w$, if $z$ satisfies $\Psi$ then $C\left(
\Psi,w,a,\sigma\right)  $ rejects with probability at least $2/3$.

(In words: $a$ leads to a \textit{sound} procedure for deciding UNSAT. \ That
is, once we find a valid $\sigma$ that's compatible with $a$, the quantum
algorithm $C$ accepts no SAT instance $\Psi$ that \textit{has} a satisfying assignment.)

\item[(4)] For all $x$, there exists a valid quantum advice state $\sigma$, as
well as a $y$ and a classical witness $w$, such that $C\left(  \urcorner
\Phi\left(  x,y\right)  ,w,a,\sigma\right)  $ accepts with probability at
least $2/3$.

(In words: $C$ verifies that for all $x$, there exists a $y$ such that
$\urcorner\Phi\left(  x,y\right)  $ is unsatisfiable. \ In other words, $C$
verifies that for all $x$, there exists a $y$ such that for all $z$, we have
$\Phi\left(  x,y,z\right)  $. \ In other words, $C$ verifies the truth of the
$\Pi_{3}^{\mathsf{P}}$-sentence $S$.)\bigskip
\end{enumerate}

As a point of clarification, whenever we quantify over quantum states (such as
$\sigma$), we can actually take a tensor product of a polynomial number of
copies of the states, as needed. \ Of course, we can't rule out the
possibility that we'll get a state that's entangled across all the registers.
\ Fortunately, though, we don't use the witness state registers in such a way
that it ever matters whether they're entangled or not.

As a second point of clarification, in forming the statement above, whenever
we have a condition that involves a quantum algorithm (say, $V$ or $C$)
accepting with probability at least $2/3$, it's implied that if the condition
fails, then the algorithm accepts with probability at most $1/3$. \ This makes
verifying the condition a quantum polynomial-time operation. \ Likewise, for
part (1), it can be guaranteed that there exists an $a$ such that, for all
$\rho_{1},\rho_{2}$ consistent with $a$ and all $\Psi$ and $w$, the difference
between the two acceptance probabilities is at most (say) $1/20$. \ In such a
case, one can verify in quantum polynomial time that the difference is at most
$1/10$.

With these clarifications, it's not hard to see that we've given an
$\mathsf{NP}^{\mathsf{NP}^{\mathsf{P{}romiseQMA}}}$ procedure. \ The
$\mathsf{NP}$ at the bottom guesses the classical advice string $a$. \ The
$\mathsf{NP}$\ in the middle guesses $\Psi$\ for part (2) and $x$\ for part
(4), and is not needed for parts (1) and (3). \ Finally, the $\mathsf{P{}%
romiseQMA}$ on top guesses the quantum advice state $\sigma$ (or $\rho
_{1},\rho_{2}$\ for part (1)), as well as $\Psi$, $w$, $y$, and $z$ as needed.
\ Crucially, quantum states are only ever guessed in the topmost,
$\mathsf{P{}romiseQMA}$\ quantifier: once guessed, they never need to be
passed on to another quantifier, which is impossible with conventional oracle calls.

But why does the procedure we've given correctly decide the $\Pi
_{3}^{\mathsf{P}}$-sentence $S$? \ Well, firstly, \textit{if} $a$ is a correct
trusted advice string, then part (4) of the procedure just directly expresses
$S$, using the assumed $\mathsf{YQ\cdot QCMA/poly}$ algorithm for
$\mathsf{coNP}$ to eliminate one of the three quantifiers in the usual manner
of Karp-Lipton theorems.

That leaves the problem of verifying that $a$ is a correct trusted advice
string. \ Parts (2) and (3) of the procedure verify the latter, by quantifying
over all possible SAT instances $\Psi$ of the appropriate polynomial size, and
checking that for each one, either $\Psi$ has a satisfying assignment or else
there's a witness $w$ that causes $C$ to accept $\Phi$, but not both. \ (In
other words, $C$ decides $\mathsf{coNP}$ in $\mathsf{YQ\cdot QCMA/poly}$.)

Now, for parts (2) and (4), we additionally needed an \textit{existential}
quantifier over the untrusted quantum advice state $\sigma$, which is then
verified using the trusted classical advice string $a$. \ The reason is that,
in parts (2) and (4), the third and final quantifier needed, over the
classical strings $y$, $z$, or $w$, happens to be existential---so that third
quantifier simply \textit{must} do \textquotedblleft double
duty\textquotedblright\ by also guessing the state $\sigma$. \ As mentioned
before, passing a quantum state from an earlier quantifier to a later one is
impossible with conventional oracle calls.

However, this need to quantify existentially over $\sigma$\ opens up a
problem. \ Namely, what if the existential quantifiers, in parts (2) or (4),
can be satisfied by \textit{different} advice states $\sigma$---states that
are all compatible with $a$, but that lead to different behaviours of $C$\ on
some inputs? \ For example, perhaps there exists an $a$ such that some
$\sigma$'s compatible with $a$ give rise to a complete verification procedure
for UNSAT, while other $\sigma$'s compatible with $a$\ give rise to a sound
verification procedure for UNSAT, but the same $\sigma$\ never gives rise to
both. \ If so, then the $\sigma$ that we find in part (4) need not give rise
to a correct $\mathsf{YQ\cdot QCMA/poly}$ algorithm for $\mathsf{coNP}$.

Fortunately, we can fix this problem using part (1). \ In part (1), we
enforced that \textit{every} state $\sigma$ compatible with $a$ must give rise
to essentially the same behaviour on every input. \ Thus, from that point
forward, it doesn't even matter whether we find $\sigma$ via a universal
quantifier or an existential one: every $\sigma$ that passes verification will
give rise to the same behaviour, and parts (2), (3), and (4) are all talking
about the same $\mathsf{YQ\cdot QCMA/poly}$ procedure that correctly decides
$\mathsf{coNP}$.
\end{proof}

\bibliography{report}

\begin{thebibliography}{10}

\bibitem{rad}
Ronald~L Rivest, Len Adleman, and Michael~L. Dertouzos.
\newblock On data banks and privacy homomorphisms.
\newblock 1978.
\newblock {\em Foundations of secure computation}, 4(11):169--180.

\bibitem{example1}
Ivan Damg{\aa}rd, Jens Groth, and Gorm Salomonsen.
\newblock {\em The Theory and Implementation of an Electronic Voting System},
  pages 77--99.
\newblock Springer US, Boston, MA, 2003.

\bibitem{example2}
Thore Graepel, Kristin Lauter, and Michael Naehrig.
\newblock {ML Confidential: Machine Learning on Encrypted Data}.
\newblock In {\em Proceedings of the 15th International Conference on
  Information Security and Cryptology}, ICISC'12, pages 1--21, Berlin,
  Heidelberg, 2013. Springer-Verlag.

\bibitem{example3}
Jo{\"e}l Alwen, Manuel Barbosa, Pooya Farshim, Rosario Gennaro, S.~Dov Gordon,
  Stefano Tessaro, and David~A. Wilson.
\newblock {\em On the Relationship between Functional Encryption, Obfuscation,
  and Fully Homomorphic Encryption}, pages 65--84.
\newblock Springer Berlin Heidelberg, Berlin, Heidelberg, 2013.

\bibitem{example4}
Sanjam Garg, Craig Gentry, Shai Halevi, Mariana Raykova, Amit Sahai, and Brent
  Waters.
\newblock Candidate indistinguishability obfuscation and functional encryption
  for all circuits.
\newblock In {\em Proceedings of the 2013 IEEE 54th Annual Symposium on
  Foundations of Computer Science}, FOCS '13, pages 40--49, Washington, DC,
  USA, 2013. IEEE Computer Society.

\bibitem{example5}
Arjan Jeckmans, Andreas Peter, and Pieter Hartel.
\newblock {\em Efficient Privacy-Enhanced Familiarity-Based Recommender
  System}, pages 400--417.
\newblock Springer Berlin Heidelberg, Berlin, Heidelberg, 2013.

\bibitem{example6}
Kristin~E. Lauter.
\newblock Practical applications of homomorphic encryption.
\newblock In {\em Proceedings of the 2012 ACM Workshop on Cloud Computing
  Security Workshop}, CCSW '12, pages 57--58, New York, NY, USA, 2012. ACM.

\bibitem{gentry}
Craig Gentry.
\newblock Fully homomorphic encryption using ideal lattices.
\newblock In {\em Proceedings of the Forty-first Annual ACM Symposium on Theory
  of Computing}, STOC '09, pages 169--178, New York, NY, USA, 2009. ACM.

\bibitem{crypto}
Jonathan Katz and Yehuda Lindell.
\newblock {\em Introduction to modern cryptography}.
\newblock CRC press, 2014.

\bibitem{bv}
Ethan Bernstein and Umesh Vazirani.
\newblock Quantum complexity theory.
\newblock {\em SIAM J. Comput.}, 26(5):1411--1473, October 1997.

\bibitem{simon}
Daniel~R. Simon.
\newblock On the power of quantum computation.
\newblock {\em SIAM Journal on Computing}, 26(5):1474--1483, 1997.

\bibitem{niel}
J~Niel De~Beaudrap, Richard Cleve, John Watrous, et~al.
\newblock Sharp quantum versus classical query complexity separations.
\newblock {\em Algorithmica}, 34(4):449--461, 2002.

\bibitem{aaronson}
Scott Aaronson and Andris Ambainis.
\newblock Forrelation: A problem that optimally separates quantum from
  classical computing.
\newblock In {\em Proceedings of the Forty-seventh Annual ACM Symposium on
  Theory of Computing}, STOC '15, pages 307--316, New York, NY, USA, 2015. ACM.

\bibitem{fitzsimons2016private}
Joseph~F Fitzsimons.
\newblock Private quantum computation: an introduction to blind quantum
  computing and related protocols.
\newblock {\em npj Quantum Information}, 3(1):23, 2017.

\bibitem{bfk}
Anne Broadbent, Joseph Fitzsimons, and Elham Kashefi.
\newblock Universal blind quantum computation.
\newblock In {\em Proceedings of the 50th Annual Symposium on Foundations of
  Computer Science}, FOCS '09, pages 517 -- 526. IEEE Computer Society, 2009.

\bibitem{Childs:2005:SAQ:2011670.2011674}
Andrew~M. Childs.
\newblock Secure assisted quantum computation.
\newblock {\em Quantum Info. Comput.}, 5(6):456--466, September 2005.

\bibitem{PhysRevA.87.050301}
Tomoyuki Morimae and Keisuke Fujii.
\newblock Blind quantum computation protocol in which alice only makes
  measurements.
\newblock {\em Phys. Rev. A}, 87:050301, May 2013.

\bibitem{PhysRevLett.111.230501}
Vittorio Giovannetti, Lorenzo Maccone, Tomoyuki Morimae, and Terry~G. Rudolph.
\newblock Efficient universal blind quantum computation.
\newblock {\em Phys. Rev. Lett.}, 111:230501, Dec 2013.

\bibitem{1607.00758}
Atul Mantri, Tommaso~F Demarie, and Joseph~F Fitzsimons.
\newblock Universality of quantum computation with cluster states and ({X},
  {Y})-plane measurements.
\newblock {\em Scientific reports}, 7:42861, 2017.

\bibitem{fk}
Joseph~F Fitzsimons and Elham Kashefi.
\newblock Unconditionally verifiable blind quantum computation.
\newblock {\em Physical Review A}, 96(1):012303, 2017.

\bibitem{PhysRevLett.111.230502}
Atul Mantri, Carlos~A. P\'erez-Delgado, and Joseph~F. Fitzsimons.
\newblock Optimal blind quantum computation.
\newblock {\em Phys. Rev. Lett.}, 111:230502, Dec 2013.

\bibitem{Morimae:2015:GSB:2871393.2871395}
Tomoyuki Morimae, Vedran Dunjko, and Elham Kashefi.
\newblock Ground state blind quantum computation on aklt state.
\newblock {\em Quantum Info. Comput.}, 15(3-4):200--234, March 2015.

\bibitem{1606.06931}
Elham Kashefi and Petros Wallden.
\newblock Garbled quantum computation.
\newblock {\em Cryptography}, 1(1):6, 2017.

\bibitem{1703.03754}
Elham Kashefi, Luka Music, and Petros Wallden.
\newblock The quantum cut-and-choose technique and quantum two-party
  computation, 2017.
\newblock Eprint:\href{http://arxiv.org/abs/1703.03754}{arXiv:1703.03754}.

\bibitem{abe}
Dorit Aharonov, Michael Ben{-}Or, and Elad Eban.
\newblock Interactive proofs for quantum computations.
\newblock In {\em Innovations in Computer Science - {ICS} 2010, Tsinghua
  University, Beijing, China, January 5-7, 2010. Proceedings}, pages 453--469,
  2010.

\bibitem{doi:10.1139/cjp-2015-0030}
Anne Broadbent.
\newblock Delegating private quantum computations.
\newblock {\em Canadian Journal of Physics}, 93(9):941--946, 2015.

\bibitem{mahadev2017classical}
Urmila Mahadev.
\newblock Classical homomorphic encryption for quantum circuits.
\newblock In {\em 2018 IEEE 59th Annual Symposium on Foundations of Computer
  Science (FOCS)}, pages 332--338. IEEE, 2018.

\bibitem{cryptoeprint:2018:338}
Zvika Brakerski.
\newblock Quantum {FHE} (almost) as secure as classical.
\newblock Cryptology ePrint Archive, Report 2018/338, 2018.
\newblock \url{https://eprint.iacr.org/2018/338}.

\bibitem{afk}
M.~Abadi, J.~Feigenbaum, and J.~Kilian.
\newblock On hiding information from an oracle.
\newblock In {\em Proceedings of the Nineteenth Annual ACM Symposium on Theory
  of Computing}, STOC '87, pages 195--203, New York, NY, USA, 1987. ACM.

\bibitem{yap}
Chee~K. Yap.
\newblock Some consequences of non-uniform conditions on uniform classes.
\newblock {\em Theoretical Computer Science}, 26(3):287 -- 300, 1983.

\bibitem{raztal}
Ran Raz and Avishay Tal.
\newblock Oracle separation of {BQP} and {PH}.
\newblock {\em eccc preprint TR18-107}, 2018.
\newblock
  Eprint:\href{https://eccc.weizmann.ac.il/report/2018/107/}{https://eccc.weizmann.ac.il/report/2018/107/}.

\bibitem{blackbox}
Thomas Jansen.
\newblock On the black-box complexity of example functions: The real jump
  function.
\newblock In {\em Proceedings of the 2015 ACM Conference on Foundations of
  Genetic Algorithms XIII}, FOGA '15, pages 16--24, New York, NY, USA, 2015.
  ACM.

\bibitem{shor}
Peter~W. Shor.
\newblock Polynomial-time algorithms for prime factorization and discrete
  logarithms on a quantum computer.
\newblock {\em SIAM Review}, 41(2):303--332, 1999.

\bibitem{aaronsonph}
Scott Aaronson.
\newblock {BQP} and the polynomial hierarchy.
\newblock In {\em Proceedings of the Forty-second ACM Symposium on Theory of
  Computing}, STOC '10, pages 141--150, New York, NY, USA, 2010. ACM.

\bibitem{bosonsampling}
Scott Aaronson and Alex Arkhipov.
\newblock The computational complexity of linear optics.
\newblock In {\em Proceedings of the Forty-third Annual ACM Symposium on Theory
  of Computing}, STOC '11, pages 333--342, New York, NY, USA, 2011. ACM.

\bibitem{toda}
Seinosuke Toda.
\newblock {PP} is as hard as the polynomial-time hierarchy.
\newblock {\em SIAM Journal on Computing}, 20(5):865--877, 1991.

\bibitem{andreas}
Andreas Bj{\"o}rklund.
\newblock {Below All Subsets for Some Permutational Counting Problems }.
\newblock In Rasmus Pagh, editor, {\em 15th Scandinavian Symposium and
  Workshops on Algorithm Theory (SWAT 2016)}, volume~53 of {\em Leibniz
  International Proceedings in Informatics (LIPIcs)}, pages 17:1--17:11,
  Dagstuhl, Germany, 2016. Schloss Dagstuhl--Leibniz-Zentrum fuer Informatik.

\bibitem{ryser1963combinatorial}
Herbert~John Ryser.
\newblock {\em Combinatorial mathematics}, volume~14.
\newblock JSTOR, 1963.

\bibitem{liao2017satellite}
Sheng-Kai Liao, Wen-Qi Cai, Wei-Yue Liu, Liang Zhang, Yang Li, Ji-Gang Ren,
  Juan Yin, Qi~Shen, Yuan Cao, Zheng-Ping Li, et~al.
\newblock Satellite-to-ground quantum key distribution.
\newblock {\em Nature}, 549(7670):43, 2017.

\bibitem{aaronson2004limits}
Scott Aaronson.
\newblock Limits on efficient computation in the physical world.
\newblock {\em arXiv preprint quant-ph/0412143}, 2004.

\bibitem{grover}
Lov~K. Grover.
\newblock A fast quantum mechanical algorithm for database search.
\newblock In {\em Proceedings of the Twenty-eighth Annual ACM Symposium on
  Theory of Computing}, STOC '96, pages 212--219, New York, NY, USA, 1996. ACM.

\bibitem{pnp}
Scott Aaronson.
\newblock $\mathsf{P} \stackrel{?}{=} \mathsf{NP}$.
\newblock 2017.
\newblock
  \href{https://www.scottaaronson.com/papers/pnp.pdf}{https://www.scottaaronson.com/papers/pnp.pdf}.

\bibitem{fhe1}
Zvika Brakerski and Vinod Vaikuntanathan.
\newblock Efficient fully homomorphic encryption from (standard) {LWE}.
\newblock In {\em Proceedings of the 2011 IEEE 52Nd Annual Symposium on
  Foundations of Computer Science}, FOCS '11, pages 97--106, Washington, DC,
  USA, 2011. IEEE Computer Society.

\bibitem{fhe2}
Zvika Brakerski, Craig Gentry, and Vinod Vaikuntanathan.
\newblock (leveled) fully homomorphic encryption without bootstrapping.
\newblock In {\em Proceedings of the 3rd Innovations in Theoretical Computer
  Science Conference}, ITCS '12, pages 309--325, New York, NY, USA, 2012. ACM.

\bibitem{fhe3}
Marten van Dijk, Craig Gentry, Shai Halevi, and Vinod Vaikuntanathan.
\newblock Fully homomorphic encryption over the integers.
\newblock In {\em Proceedings of the 29th Annual International Conference on
  Theory and Applications of Cryptographic Techniques}, EUROCRYPT'10, pages
  24--43, Berlin, Heidelberg, 2010. Springer-Verlag.

\bibitem{doi:10.1142/S0219749906002171}
Pablo Arrighi and Louis Salvail.
\newblock Blind quantum computation.
\newblock {\em International Journal of Quantum Information}, 04(05):883--898,
  2006.

\bibitem{broadbent}
Anne Broadbent and Stacey Jeffery.
\newblock Quantum homomorphic encryption for circuits of low {T}-gate
  complexity.
\newblock In {\em Advances in Cryptology - {CRYPTO} 2015 - 35th Annual
  Cryptology Conference, Santa Barbara, CA, USA, August 16-20, 2015,
  Proceedings, Part {II}}, pages 609--629, 2015.

\bibitem{dulek}
Yfke Dulek, Christian Schaffner, and Florian Speelman.
\newblock {\em Quantum Homomorphic Encryption for Polynomial-Sized Circuits},
  pages 3--32.
\newblock Springer Berlin Heidelberg, Berlin, Heidelberg, 2016.

\bibitem{Alagic2016}
Gorjan Alagic, Anne Broadbent, Bill Fefferman, Tommaso Gagliardoni, Christian
  Schaffner, and Michael St.~Jules.
\newblock {\em Computational Security of Quantum Encryption}, pages 47--71.
\newblock Springer International Publishing, Cham, 2016.

\bibitem{qfheimposs}
Li~Yu, Carlos~A. P\'erez-Delgado, and Joseph~F. Fitzsimons.
\newblock Limitations on information-theoretically-secure quantum homomorphic
  encryption.
\newblock {\em Phys. Rev. A}, 90:050303, Nov 2014.

\bibitem{newmannshi}
Michael Newman and Yaoyun Shi.
\newblock Limitations on transversal computation through quantum homomorphic
  encryption.
\newblock {\em Quantum Information and Computation}, 18:0927--0948, 2018.

\bibitem{Morimae}
Tomoyuki Morimae and Takeshi Koshiba.
\newblock Impossibility of perfectly-secure delegated quantum computing for
  classical client, 2014.
\newblock Eprint:\href{http://arxiv.org/abs/1407.1636}{arXiv:1407.1636}.

\bibitem{vedran}
Vedran Dunjko and Elham Kashefi.
\newblock Blind quantum computing with two almost identical states, 2016.
\newblock Eprint:\href{http://arxiv.org/abs/1604.01586}{arXiv:1604.01586}.

\bibitem{harrow2017quantum}
Aram~W Harrow and Ashley Montanaro.
\newblock Quantum computational supremacy.
\newblock {\em Nature}, 549(7671):203, 2017.

\bibitem{nc}
Michael~A. Nielsen and Isaac~L. Chuang.
\newblock {\em Quantum Computation and Quantum Information: 10th Anniversary
  Edition}.
\newblock Cambridge University Press, New York, NY, USA, 10th edition, 2011.

\bibitem{watrous2009quantum}
John Watrous.
\newblock Quantum computational complexity.
\newblock In {\em Encyclopedia of complexity and systems science}, pages
  7174--7201. Springer, 2009.

\bibitem{zoo}
Complexity {Z}oo.
\newblock \url{https://complexityzoo.uwaterloo.ca/Complexity_Zoo}.

\bibitem{marpoly}
Scott Aaronson.
\newblock {QMA/qpoly $\subseteq$ PSPACE/poly:} de-merlinizing quantum
  protocols.
\newblock In {\em in Proceedings of 21st IEEE Conference on Computational
  Complexity}, 2006.

\bibitem{stockmeyer}
Larry Stockmeyer.
\newblock The complexity of approximate counting.
\newblock In {\em Proceedings of the fifteenth annual ACM symposium on Theory
  of computing}, pages 118--126. ACM, 1983.

\bibitem{bppph}
Clemens Lautemann.
\newblock {$\mathsf{BPP}$} and the polynomial hierarchy.
\newblock {\em Information Processing Letters}, 17(4):215--217, 1983.

\bibitem{brassard}
G.~Brassard.
\newblock A note on the complexity of cryptography (corresp.).
\newblock {\em IEEE Transactions on Information Theory}, 25(2):232--233, Mar
  1979.

\bibitem{adleman}
Leonard Adleman.
\newblock Two theorems on random polynomial time.
\newblock In {\em Proceedings of the 19th Annual Symposium on Foundations of
  Computer Science}, SFCS '78, pages 75--83, Washington, DC, USA, 1978. IEEE
  Computer Society.

\bibitem{qrfs}
Scott Aaronson.
\newblock Quantum lower bound for recursive {F}ourier sampling.
\newblock {\em Quantum Info. Comput.}, 3(2):165--174, March 2003.

\bibitem{watroussep}
J.~Watrous.
\newblock Succinct quantum proofs for properties of finite groups.
\newblock In {\em Proceedings of the 41st Annual Symposium on Foundations of
  Computer Science}, FOCS '00, pages 537--, Washington, DC, USA, 2000. IEEE
  Computer Society.

\bibitem{ak}
Scott Aaronson and Greg Kuperberg.
\newblock Quantum versus classical proofs and advice.
\newblock In {\em Computational Complexity, 2007. CCC'07. Twenty-Second Annual
  IEEE Conference on}, pages 115--128. IEEE, 2007.

\bibitem{kp}
Richard~M. Karp and Richard~J. Lipton.
\newblock Turing machines that take advice.
\newblock {\em L'Enseignement Math\'{e}matique}, 28:191--201, 1982.

\bibitem{qkp}
Scott Aaronson and Andrew Drucker.
\newblock A full characterization of quantum advice.
\newblock In {\em Proceedings of the Forty-second ACM Symposium on Theory of
  Computing}, STOC '10, pages 131--140, New York, NY, USA, 2010. ACM.

\bibitem{aaronsondrucker}
Scott Aaronson and Andrew Drucker.
\newblock A full characterization of quantum advice, 2010.
\newblock Eprint:\href{http://arxiv.org/abs/1004.0377}{arXiv:1004.0377}.

\end{thebibliography}
\bibliographystyle{unsrt}

\end{document}